\documentclass[letterpaper,journal]{IEEEtran}

\usepackage[english]{babel}
\usepackage[utf8]{inputenc}
\usepackage[T1]{fontenc}

\usepackage{graphicx}
\usepackage{xcolor}
\usepackage[caption=false,font=normalsize,labelfont=sf,textfont=sf]{subfig}

\makeatletter
\let\IEEEorig@makecaption\@makecaption
\def\@IEEEfigurestring{figure}
\long\def\@makecaption#1#2{%
  \ifx\@captype\@IEEEfigurestring
    \@IEEEfigurecaptionsepspace
    \setbox\@tempboxa\hbox{\normalfont\small {#1.}\nobreakspace\nobreakspace #2}%
    \ifdim \wd\@tempboxa >\hsize
      \setbox\@tempboxa\hbox{\normalfont\small {#1.}\nobreakspace\nobreakspace}%
      \parbox[t]{\hsize}{\normalfont\small\noindent\unhbox\@tempboxa#2}%
    \else
      \ifCLASSOPTIONconference
        \hbox to\hsize{\normalfont\small\hfil\box\@tempboxa\hfil}%
      \else
        \hbox to\hsize{\normalfont\small\box\@tempboxa\hfil}%
      \fi
    \fi
  \else
    \IEEEorig@makecaption{#1}{#2}%
  \fi
}
\makeatother

\usepackage{amsmath,amssymb,amsthm,mathrsfs,bbm,mathbbol}
\usepackage{mathabx} %
\usepackage{scalerel,stackengine}

\usepackage{enumitem}
\usepackage{booktabs}
\usepackage{multirow}
\usepackage{makecell}

\usepackage{algorithm}
\usepackage{algorithmic}
\usepackage{url}  
\usepackage{cite}
\usepackage[hidelinks]{hyperref}

\newcommand{\coloneqq}{\mathrel{\mathop:}=}
\DeclareSymbolFontAlphabet{\amsmathbb}{AMSb}%

\newcommand{\cp}[1]{\ifmmode {\mathcal{#1}}\else ${\mathcal{#1}}$\fi}

\newcommand{\bW}{\mathbf{W}}
\newcommand{\bX}{\mathbf{X}}

\newcommand{\bb}{\mathbf{b}}

\newcommand{\bu}{\mathbf{u}}

\newcommand{\bx}{\mathbf{x}}
\newcommand{\bw}{\mathbf{w}}
\newcommand{\bz}{\mathbf{z}}

\newcommand{\calC}{\mathcal{C}}

\newcommand{\calX}{\mathcal{X}}

\newcommand{\bbE}{\mathbb{E}}

\newcommand{\bbQ}{\mathbb{Q}}

\newcommand\tensor[1]{%
  \ifcat\noexpand#1\relax %
    \mathbb{#1}%
  \else
      \if\relax\detokenize\expandafter{\romannumeral-0#1}\relax  %
        \mathbb{#1}
      \else
        \boldsymbol{\mathcal{#1}}%
      \fi
  \fi }

\definecolor{darkgreen}{rgb}{0., 0.4, 0.}
\definecolor{darkorange}{rgb}{1.0, 0.55, 0.0}

\newtheorem{theorem}{Theorem}%

\newtheorem{proposition}{Proposition}

\newtheorem{definition}{Definition}
\newtheorem{remark}{Remark}

\newtheorem{example}{Example}

\stackMath

\newcommand\mywidehat[1]{%
\savestack{\tmpbox}{\stretchto{%
  \scaleto{%
    \scalerel*[\widthof{\ensuremath{#1}}]{\kern3pt\mathchar"0362\kern3pt}%
    {\rule{0ex}{\textheight}}%
  }{\textheight}%
}{2.4ex}}%
\stackon[-7.2pt]{#1}{\tmpbox}{}%
}

\hypersetup{
  pdftitle={Group-invariant Coresets for Data-efficient Active Learning},
  pdfauthor={Luciano C. Ayres, Jose C. M. Bermudez, Sergio J. M. de Almeida, and Ricardo A. Borsoi},
  pdfsubject={Preprint submitted to IEEE Transactions on Neural Networks and Learning Systems},
  pdfkeywords={active learning, coresets, invariance, quotient space}
}
\title{Group-invariant Coresets for Data-efficient Active Learning}
\author{Luciano C. Ayres, Jos\'{e} C. M. Bermudez, S\'{e}rgio J. M. de Almeida, and Ricardo A. Borsoi%
\thanks{L. C. Ayres and J. C. M. Bermudez are with Universidade Federal de Santa Catarina, Florian\'{o}polis, SC, Brazil; S. J. M. de Almeida is with Universidade Cat\'{o}lica de Pelotas, Pelotas, RS, Brazil; R. A. Borsoi is with Universit\'{e} de Lorraine, CNRS, CRAN, Vandoeuvre-l\`{e}s-Nancy, France. E-mails: lucayress@gmail.com, j.bermudez@ieee.org, sergio.almeida@ucpel.edu.br, and ricardo.borsoi@univ-lorraine.fr. Corresponding author: L. C. Ayres.}%
\thanks{This work has been submitted to the IEEE for possible publication. Copyright may be transferred without notice, after which this version may no longer be accessible.}}
\begin{document}

\maketitle

\begin{abstract}
Active learning reduces labeling cost by querying the most informative unlabeled samples, but standard coreset methods ignore known data symmetries and can waste budget on transformed versions of the same instance. We propose GRINCO, a group-invariant coreset framework that performs acquisition in the quotient space induced by a transformation group, so that selection operates on orbits rather than raw samples. The method uses either canonical representatives or learned orbit-separating invariant embeddings to define practical quotient metrics, and combines quotient-space $k$-center selection with invariant training through an orbit-averaged loss. We further derive a generalization bound that relates excess orbit-averaged risk to quotient-space coverage, label uncertainty, and intra-orbit variability. Experiments on synthetic scale-invariant data and image benchmarks with rotation-induced redundancy show that GRINCO improves orbit coverage and achieves stronger label efficiency than conventional coreset baselines, especially when group-induced redundancy is substantial.
\end{abstract}


\section{Introduction}

\IEEEPARstart{A}{cquiring} high-quality labels for classification tasks can be expensive and time-consuming, particularly in domains such as medical imaging~\cite{zhang2023learning}, precision agriculture~\cite{li2023label}, and remote sensing~\cite{patel2023comprehensive}.
This has made data efficiency a key objective in machine learning~\cite{gui2024sslsurvey,lin2025hslabeling}.
\emph{Active Learning} (AL) aims to reduce labeling burden by selecting only the most informative samples for annotation~\cite{settles2009}. In contrast to standard supervised learning, AL adaptively queries data in a \emph{pool} of unlabeled samples which are expected to provide the greatest utility for model improvement. By iterating between model training and targeted querying, active strategies can often reach state-of-the-art performance with substantially fewer labeled examples~\cite{li2024survey}.

AL frequently relies on diversity/coverage criteria to avoid redundant queries. One of the most effective techniques is the use of \textit{coresets}~\cite{agarwal2005geometric}, which construct a small but representative subset of the dataset~\cite{sener2018coresets,moser2025coreset}. 
The fundamental design is based on geometric covering objectives such as the $k$-center objective, which selects a subset of $k$ samples that best covers the unlabeled pool~\cite{sener2018coresets}.
This approach has later been refined with criteria that aim to favor locally sparse regions of the pool~\cite{kim2022dacs} or to improve robustness in low-budget regimes~\cite{yehuda2022covering,bae2024generalized}. 
Recent work also considered the use of fairness constraints~\cite{xiong2024fair}, robust metrics~\cite{acharya2025geometric} as well as statistical models of repulsive inter-sample interactions~\cite{bardenet2024ddps} to improve different aspects of coreset selection (see Section~\ref{subsec:coresets} for more details).

Despite their success, coreset methods suffer from a key limitation: existing works do not properly account for data invariances (i.e., transformations of the input that preserve semantic content). This can lead to the selection of redundant samples that are transformed versions of the same datum. 
On the other hand, significant effort has been dedicated to accounting for data symmetries in supervised and representation learning using the framework of \emph{group theory} (see Section~\ref{subsec:invarianceGeneral}), which improves sample efficiency and generalization~\cite{weiler2024gauge}.
This can be performed by designing neural network architectures that are intrinsically invariant to the action of a known symmetry group~\cite{cohen2016group,worrall2017harmonic,finzi2020generalizing,zaheer2017deepsets,bronstein2021geometric}.

A different approach is to enforce invariance through carefully designed data augmentation~\cite{chen2020group}, or by the use of approaches which map transformed versions of the same sample to a canonical representative~\cite{aslan2023group}.
Moreover, recent self-supervised representation learning (SSL) methods such as RotNet~\cite{gidaris2018unsupervised}, SimCLR~\cite{chen2020simclr} and DINO~\cite{caron2021dino} learn embeddings that are empirically robust to augmentation-induced nuisance factors due to carefully designed training criteria, including work that learns invariances to a specific group directly from training data to yield invariant embeddings~\cite{benton2020learning}.
However, these approaches primarily address model and representation learning rather than data acquisition and coreset selection.

Recent works have shown the importance of employing good representations to achieve competitive performance in coreset selection in AL. 
For instance,~\cite{bengar2021reducing} showed that using self-supervised learning (SSL) to improve the representation can have a larger impact on AL performance than the choice of sophisticated query strategies.
Similarly,~\cite{doucet2024bridging} investigated hybrid strategies that combine diversity-based selection in early query rounds followed by uncertainty-based sampling, along with SSL to improve the learned models.
Deep metric learning has also been used in AL for coreset selection in 3D image segmentation~\cite{vepa2024integrating} to obtain a diversity objective where the distances better reflect task-relevant similarity.
Another AL algorithm targeting 3D molecular graphs has incorporated graph isomorphism and isometries in the designed querying rule to avoid selecting redundant queries under physically meaningful molecular graph transformations~\cite{subedi2024empowering}.

Despite the effectiveness of the aforementioned strategies in promoting sample diversity in coreset selection and AL, they often overlook inherent data symmetries, which can lead to the selection of redundant (transformed) samples.
In this paper, we address this gap by proposing a general \emph{group-invariant coreset construction framework} (which we call \emph{GRINCO}) for AL that performs acquisition at the level of \emph{orbits} induced by a known transformation group (i.e., the set of admissible group-transformed versions of a sample)~\cite{armstrong1997groups,weiler2024gauge}. By designing a coverage criterion in the corresponding \emph{quotient space} induced by the group action (see Section~\ref{subsec:invarianceGeneral}), the proposed method enforces non-redundancy at the orbit level. 

Candidate pool elements are effectively organized into orbits and label budget is allocated to semantically distinct samples, promoting data efficiency. Thus, our GRINCO coreset framework can achieve state-of-the-art performance in AL with fewer labeled examples.
Unlike previous works that improved representations used in specific application tasks (e.g., molecular graphs or 3D segmentation)~\cite{bengar2021reducing,subedi2024empowering,vepa2024integrating}, our framework provides a general coreset formulation that is not tied to a specific domain or task, and can be instantiated for different transformation groups and application domains.
Moreover, by leveraging the concept of \emph{orbit-separating functions} \cite{nguyen2023learningSymmetrizationOrbitDistance,dym2024lowDimensionalInvariantEmbeddings} (mappings that distinguish orbits), we introduce quotient metrics that can be implemented using representation learning to compute distances between orbits. 
This leads to a principled and computationally efficient approach for coreset selection in the quotient space applicable to different groups.

Furthermore, our coreset acquisition criterion uses a symmetry group that can be aligned with the invariances accounted for in the training of the downstream model, which can be achieved via invariant architectures or data-augmentation strategies~\cite{chen2020group,cohen2016group,finzi2020generalizing}.
A theoretical analysis for the proposed method yields generalization bounds that link the orbit-averaged risk (i.e., the expected loss averaged over the group orbit of each sample) to coverage in the quotient metric, while also accounting for label noise and intra-orbit variability.
Experimental results show that our approach achieves improved performance compared to state-of-the-art coreset-based AL methods, attaining comparable accuracy with fewer labeled examples, particularly in datasets that contain significant group-induced redundancy.
The contributions of this work include:

\begin{itemize}
    \item A group-theoretic coreset formulation that performs acquisition over orbits in the quotient space induced by the transformation group, ensuring a coverage criterion that inherently excludes redundant copies of the same underlying sample.
    \item An efficient instantiation of the framework via quotient-space metrics (distances between orbits) induced by orbit-separating mappings (analytic canonicalization or learned embedders), together with a $k$-centers-style selection procedure in the induced quotient metric coupled with a training objective based on orbit-averaged losses.

    \item A theoretical analysis linking coverage in the quotient metric to excess orbit-averaged risk. This yields generalization bounds and a term-wise interpretation that clarifies the impact of label noise, intra-orbit variability, and group-prior misspecification on the resulting error.
\end{itemize}

The remainder of this paper is structured as follows. Section~\ref{sec:background_review} reviews the background and related work. Section~\ref{sec:proposed_approach} introduces our proposed framework, including the group invariant coreset formulation, the corresponding AL framework, and its theoretical analysis. Section~\ref{sec:experiments} presents the experimental results. Section~\ref{sec:conclusion} concludes the paper. Appendix~\ref{app:generalization-proof} provides the proof of the generalization theorem.

\section{Background on AL, coresets and group theory}\label{sec:background_review}

Let $\mathcal{X} \subseteq \mathbb{R}^d$ denote the input space and $\mathcal{Y} = \{1,2,\dots,C\}$ the label set.
We consider a labeled dataset $\mathcal{D}_L = \{(\bx_i, y_i)\}_{i=1}^{N_L}$, where each feature vector $\bx_i \in \mathcal{X}$ is associated with a class label $y_i \in \mathcal{Y}$.
A parametric model $f_{\bw}: \mathcal{X} \to \Delta_C$, with parameter vector $\bw \in \mathbb{R}^m$, where $\Delta_C := \{\boldsymbol{p} \in \mathbb{R}^C : p_c \ge 0,\; \sum_{c=1}^C p_c = 1\}$ is the $C$-class probability simplex, is learned by minimizing a pointwise loss $\mathcal{L}(f_{\bw}(\bx), y)$ (e.g., cross-entropy). 
Here $f_{\bw}(\bx) \in \Delta_C$ is the model output (class-probability vector) for input $\bx \in \mathcal{X}$ and $y \in \mathcal{Y}$ is the corresponding true label.

Assuming $(\bx_i, y_i)$ are drawn i.i.d. from a joint measure $p(\bX,Y)$ over $\mathcal{X} \times \mathcal{Y}$, the empirical risk is denoted by
\begin{equation}
    \hat{\mathscr{R}}(f_{\bw}) 
    \;=\;
    \frac{1}{N_L}
    \sum_{i=1}^{N_L}
    \mathcal{L}(f_{\bw}(\bx_i), y_i).
\end{equation}
This constitutes an approximation to the \emph{population risk} $\mathscr{R}(f_{\bw}) =\bbE_{p(\bX,Y)}\!\big[\mathcal{L}(f_{\bw}(\bX), Y)\big]$,
where $\bbE_{p(\bX,Y)}$ denotes expectation with respect to the measure $p(\bX,Y)$.

\subsection{Active Learning}
In AL, an additional unlabeled dataset $\mathcal{D}_U = \{\bx_i\}_{i=N_L+1}^{N_L+N_U}$ with $N_U$ samples is also available. 
Many tasks demand substantial labeled data for robust classification \cite{zhang2023learning,zhao2023autonomous}. AL alleviates this burden by iteratively querying labels for the most informative samples in $\mathcal{D}_U$ \cite{cohn1994improving}.
In a pool-based scenario \cite{settles2009}, one starts with a small labeled set $\mathcal{D}_L$ and a large unlabeled pool $\mathcal{D}_U$. A query strategy $\mathcal{A}$ identifies the point $\bx^* \in \mathcal{D}_U$ that maximizes an informativeness criterion. An oracle $\mathcal{O}$ then provides the label $y^*$, after which $(\bx^*, y^*)$ is added to $\mathcal{D}_L$ and the model is retrained. This process continues until a labeling budget is exhausted or desired performance is reached \cite{sun2010survey,ren2021survey}.

In practice, one often queries a full batch of samples for labeling at each iteration rather than a single point, as this reduces retraining overhead and leverages parallel labeling.
Among various query strategies, such as uncertainty sampling, query-by-committee, expected model change, variance reduction, and density-weighted methods, coreset approaches have been recently emphasized for selecting points that comprehensively represent the data and improving efficiency \cite{kumar2020active,tharwat2023survey}.
A more detailed definition of the batch AL selection, labeling, dataset-update, and retraining will be presented in Section~\ref{subsec:invar-batch-al} in the context of the proposed method.

\subsection{Coresets} \label{subsec:coresets}

Given a dataset $\mathcal{D}=\{\bx_i\}_{i=1}^N$ with $N$ samples, a \emph{coreset} $\mathcal{C}\subseteq \mathcal{D}$ of size $|\mathcal{C}|\ll N$ is a small representative subset of the original dataset such that solving some optimization problem using only the samples in $\mathcal{C}$ closely approximates the solution that would be obtained on the full dataset~\cite{agarwal2005geometric,feldman2020core}. Coresets are also equipped with weights $w_{\bx}>0$ which reflect the importance of each datapoint $\bx \in \mathcal{C}$.
Originally introduced to tackle geometric approximation tasks like $k$-means and $k$-median clustering \cite{agarwal2005geometric}, coresets have been applied to regression, classification, and Bayesian inference problems \cite{feldman2020core}. 

In AL, coreset-based approaches select samples from $\mathcal{D}_U$ that maximize coverage and diversity, leading to efficient querying strategies that reduce labeling redundancy \cite{sener2018coresets}. 
Of particular interest are geometric coresets, which are constructed based only on the input data $\bx$. A classic example is the $k$-centers formulation, which aims to select $k$ representative centers such that every other point in $\mathcal{D}_U$ is as close as possible to its nearest center \cite{Gonzalez1985}. 
This objective (which is NP-hard) is commonly approximated by a greedy procedure that iteratively adds the point in $\mathcal{D}_U$ which is the farthest from the coreset $\mathcal{C}$ that was selected in the previous iteration. This prevents repeatedly querying samples from densely populated regions.
Moreover, no label or model-specific information is required, making these coresets agnostic to the subsequent task.

Recent works have proposed other coreset extensions to address specific challenges, including fairness, robustness, and metric learning.
For instance, \cite{xiong2024fair} proposes a framework that generates weighted synthetic samples by minimizing the Wasserstein distance between the coreset and the original dataset, subject to demographic parity (fairness) constraints ignored by standard geometric approaches.
Probabilistic sampling methods based on determinantal point processes (which model repulsive interactions) have also been explored for coreset construction~\cite{bardenet2024ddps}.
Representation learning has also been integrated in subset selection for high-dimensional data to find better metrics when the raw input distances are uninformative~\cite{vepa2024integrating}.
The sensitivity of empirical mean-based coreset selection to noisy or corrupted data was also addressed in~\cite{acharya2025geometric} by using the geometric median, which offers a higher breakdown point against outliers. 
However, these methods can be computationally intensive (especially for large datasets) and none of these works properly address group-induced redundancy via invariance-aware coreset selection, which is the focus of our proposed method.

\subsection{Group-invariance in statistical learning} \label{subsec:invarianceGeneral}

To provide the necessary background for our invariant-aware framework, we briefly review the relevant group-theoretic terminology \cite{armstrong1997groups,weiler2024gauge}. We focus on the group-theoretic concepts that enable us to rigorously define equivalence relations among transformed samples, and conclude by linking these concepts to statistical learning and risk minimization.

\subsubsection{Groups, Actions, and Orbits}\label{subsec:group-theory-basics}

A \emph{group} is a pair $(G,\cdot)$ where $G$ is a set with a binary operation that satisfies closure, associativity, has an identity element $e$, and provides an inverse for every element.
A group $G$ is finite if $|G|<\infty$. If the multiplication and inversion are smooth maps, $G$ is said to be a \emph{Lie group}.
Given a set (data space) $\mathcal{X}$, a \emph{left action} of $G$ on $\mathcal{X}$ is a map
\begin{equation}
G\times \mathcal{X}\to \mathcal{X},\quad (g,\bx)\mapsto g\cdot \bx
\end{equation}
such that $e\cdot\bx=\bx$ and $(gh)\cdot\bx=g\cdot(h\cdot\bx)$ for all transformations $g,h\in G$, $\bx\in \mathcal{X}$.
If $\mathcal{X}$ is a vector space, a (linear) representation is a homomorphism $\rho_{\mathcal{X}}:G\to GL(\mathcal{X})$ so that $g\cdot \bx=\rho_{\mathcal{X}}(g)\,\bx$.
A map $f:\mathcal{X}\to \mathcal{Y}$ is \emph{$G$-invariant} if $f(g\cdot \bx)=f(\bx)$ for all $g\in G$. 
The \emph{orbit} of $\bx\in \mathcal{X}$ under a group $G$ is
\begin{equation}
\mathscr{O}_G(\bx) := G\cdot \bx=\{g\cdot \bx:\ g\in G\},
\end{equation}
i.e., all transformations of $\bx$ under $G$. The stabilizer (isotropy subgroup) of $\bx$ is $ G_{\bx} := \{g\in G:\ g\cdot \bx = \bx\} $.
When $G$ is finite, the orbit-stabilizer theorem gives $|\mathscr{O}_G(\bx)| = |G|/|G_{\bx}|$.
The \emph{quotient space} $\bar{\mathcal{X}} \coloneqq \mathcal{X}/G$ is the set of all orbits, with orbit map $\pi_G:\mathcal{X}\to \bar{\mathcal{X}}$, $\bx\mapsto \mathscr{O}_G(\bx)$. We use $\mathscr{O}_G(\bx)$ consistently to denote an orbit, and we use $[\bx]_G$ only as shorthand when needed.
When a canonical representative of an orbit is well defined, we denote it by $\bar{\bx}\in\mathscr{O}_G(\bx)$.

Two brief illustrative examples follow. For rotations, let $G=SO(2)$ act on images $\mathcal{X}=\mathbb{R}^{H\times W}$ by in-plane rotation about the image center. The orbit $\mathscr{O}_G(\bx)=\{g \cdot \bx: g \in SO(2)\}$ collects all rotated versions of $\bx$. Note that a class label is typically invariant to this action.
For uniform scalings, let $G=\mathbb{R}_{+}$ (under multiplication) act elementwise by $g\cdot \bx = g\,\bx$.
The orbit $\mathscr{O}_G(\bx)=\{g\,\bx: g\in \mathbb{R}_{+}\}$ contains brightness or intensity rescalings of the same sample, for example due to illumination or exposure changes.
These concrete group actions show how multiple transformed inputs can represent a single semantic instance. In both cases, reasoning on the quotient space $\mathcal{X}/G$ avoids redundancy and enhances sample efficiency by effectively unifying all transformed inputs into equivalence classes.

\subsubsection{Use in statistical learning}\label{subsec:stat-use}

From a statistical learning standpoint, a symmetry group $G$ often represents transformations under which the semantic content of $\bx$ is unchanged. 
This is related to probabilistic symmetries \cite{bloemreddy2020probabilistic}, where invariance is imposed at the level of conditional distributions. Recent work also demonstrated how invariance can improve statistical efficiency by means of data augmentation, which can reduce the variance of estimators under invariant distributions \cite{chen2020group}.
Besides data augmentation \cite{chen2020group}, other established methods for handling symmetries in statistical learning include explicit regularization and invariance-aware representation learning~\cite{shorten2019survey,marchetti2023equivariant,anselmi2019symmetry,shakerinava2022structuring}, as well as neural network architecture designs that guarantee invariance~\cite{cohen2016group,zaheer2017deepsets,bronstein2021geometric,weiler2024gauge}.

Formally, a classifier (or model) $f:\mathcal{X}\to\Delta_C$ is \emph{$G$-invariant} if it assigns the same outputs to all transformed inputs, that is,
$f(\bx) = f(g\cdot \bx)$ for all $g\in G$.
Consequently, the assigned labels are invariant under the group action. 
More generally, a statistical learning problem is $G$-invariant when the conditional distribution of outputs remains unchanged under group actions:
\begin{definition}[Label invariance]\label{def:label-invariance}
We say that $p(Y\mid\bx)$ is $G$-invariant if, for all $g\in G$,
\begin{equation} \label{eq:label-invariance-eq}
    p(Y\mid\bx) = p\bigl(Y\mid g\cdot \bx\bigr). \end{equation}
\end{definition}
\begin{remark} 
Deterministic classifiers are included by taking $p(Y\mid\bx) = \delta_{Y-\hat{y}(\bx)}$. Here $\delta_{a-b}$ equals $1$ if $a=b$ and $0$ otherwise, and $\hat{y}(\bx)$ is the predicted label for input $\bx$.
\end{remark}

Under the label invariance assumption, data augmentation can be mathematically modeled by sampling transformations $g$ from a probability measure $\bbQ$ supported on the group $G$, and training on transformed pairs $(g \cdot \bx_i, y_i)$, for $g\sim \bbQ$ and $(\bx_i, y_i)\sim p(\bX,Y)$. When $G$ is compact, $\bbQ$ can be taken to be the normalized Haar measure on $G$, i.e., the unique measure satisfying left and right invariance, $\bbQ(gA)=\bbQ(A)$ and $\bbQ(Ag)=\bbQ(A)$ for all $g\in G$ and measurable $A\subseteq G$.
Sampling $g\sim\bbQ$ corresponds to drawing a group transformation ``uniformly''.

For $G=SO(2)$, $\bbQ$ is the uniform distribution over angles. For finite groups such as the cyclic rotation group $C_4=\{0^\circ,90^\circ,180^\circ,270^\circ\}$, it is the discrete uniform measure. When $G$ is non-compact, such as arbitrary scalings, no finite Haar probability measure exists. In such cases, one can either consider the Haar measure on a compact subgroup, or define an application-specific measure that reflects the desired variability.
This probabilistic view makes augmentation an expectation over group actions under $\bbQ$, which motivates the \emph{per-sample orbit-averaged loss}~\cite{chen2020group}:
\begin{align}
    \mathcal{L}^{\bbQ}(\bx, y, f_{\bw})
    & := \bbE_{\bbQ(g)}\bigl[\mathcal{L}\bigl(f_{\bw}(g\cdot \bx), y\bigr)\bigr]
    \nonumber \\
    & \,= \int_G \mathcal{L}\bigl(f_{\bw}(g\cdot \bx), y\bigr)\, d\bbQ(g),
    \label{eq:orbit-averaged-loss}
\end{align}
which can be interpreted as averaging along the orbit of $\bx$ induced by $G$. In implementations, 
this expectation is approximated via a Monte Carlo average over sampled transformations.
This group-based perspective will be instrumental for our construction of group-invariant coresets in quotient spaces, where samples differing only by a transformation in $G$ are treated under a single equivalence class.

\section{Proposed approach: group-invariant coresets and AL for data efficiency}\label{sec:proposed_approach}

The fundamental idea of our \emph{group-invariant coreset} framework is to account for symmetries in the data (induced by a group $G$) during the selection process. This improves data efficiency compared to standard criteria~\cite{Gonzalez1985,sener2018coresets} when redundant samples are present.
To this end, we redefine the unit of selection to a geometric object in the form of \emph{orbits} $\mathscr{O}_G(\bx)$, which naturally incorporates group information.

\begin{figure}[t]
    \centering
    \includegraphics[width=1.00\linewidth]{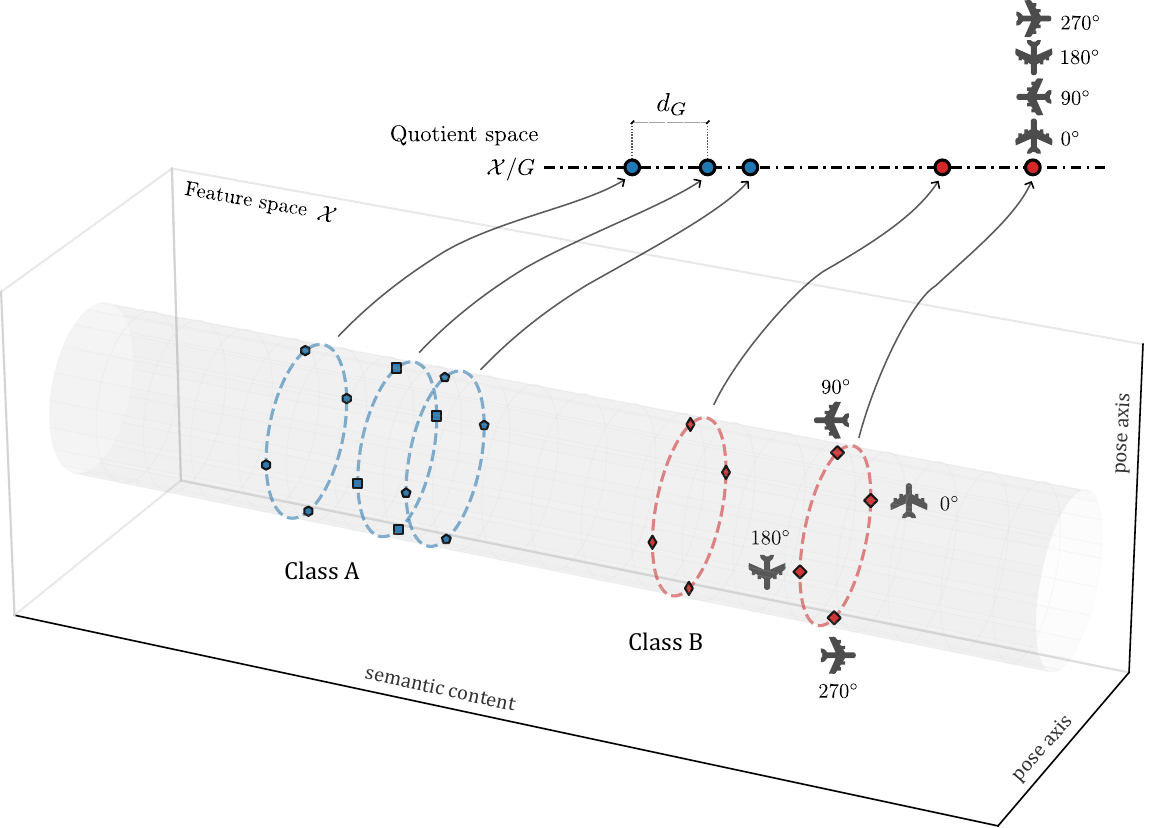}
    \caption{Illustrative example. In the input space $\calX$, one axis corresponds to semantic content (i.e., distinct classes) and the other two to pose information (rotation angle) induced by a group $G$. Each point $\bx\in\calX$ (such as an airplane) generates an orbit $\mathscr{O}_G(\bx)$ as it undergoes transformations by $G$ (shown as colored dashed loops). The quotient space $\calX/G$ captures only the semantic content, mapping each orbit to a single point. A quotient metric $d_G$ measures a distance between two orbits, being invariant to pose information.}
    \label{fig:concept-invariant}
\end{figure}

Figure~\ref{fig:concept-invariant} illustrates this idea with a simple example based on the rotation group $G=SO(2)$, with the input space $\calX$ represented as a cylinder in 3D space, one axis corresponding to semantic content (defining the class) and the other two to pose information (rotation angle). Each point $\bx\in\calX$ (e.g., an airplane) generates a continuous trajectory of points as it undergoes transformations by $G$, which corresponds to the orbit $\mathscr{O}_G(\bx) = \{g \cdot \bx : g \in G\}\subset\calX$ containing the set of transformed examples (shown as colored dashed loops). 
The pose axes can be ``collapsed'' by mapping the input space $\calX$ to the \emph{quotient space} $\calX/G$, which captures only the semantic content. In this representation, each orbit $\mathscr{O}_G(\bx)$ is mapped to a single point in $\calX/G$ by a projection~$\pi_G$.

Considering a standard distance-based selection criterion operating in $\calX$ immediately reveals potential issues, as an airplane at $0^\circ$ and its $180^\circ$ rotation might be treated as far from one another despite being semantically identical.
Therefore, we consider a \emph{quotient metric} $d_G$ that measures the distance between two orbits, effectively formulating the selection criterion in the quotient space $\calX/G$. In the illustration, this corresponds to the distance between loops, which captures semantic content while being invariant to pose alignment. This ensures that the AL oracle needs to supply only a single label per orbit, thereby concentrating the labeling budget on task-related information.
This will be formalized in detail in Section~\ref{subsec:group-invariant-coresets}.

Following this idea, the proposed selection objective is a coreset $\mathcal{C}_{\calX/G}$ on the quotient-space, which we call \textbf{GRINCO} (standing for GRoup-INvariant COreset).
Denoting by $\mathcal{D}_U=\{\bx_n\}_{n=1}^{N_U}$ the unlabeled pool\footnote{We index the samples $\bx_n$ starting from $n=1$ for convenience.}, each sample $\bx_i$ induces an orbit $\mathscr{O}_i \coloneqq \mathscr{O}_G(\bx_i)$ under $G$. 
The key idea behind GRINCO is that the coreset is constructed over orbits rather than in the input space, thus, by selecting a set of representative indices $i_1,\dots,i_K\in\{1,\dots,N_U\}$, the coreset can be expressed as:
\begin{equation}
    \mathcal{C}_{\calX/G} = \{\mathscr{O}_{i_1}, \dots, \mathscr{O}_{i_K}\} \,,
\end{equation}
where $K\ll N_U$ is the coreset size. Note that, in general, the set of indices $i_1,\dots,i_K$ is not unique: if two samples $\bx_n$ and $\bx_m$ belong to the same orbit, then $\mathscr{O}_{i_n}=\mathscr{O}_{i_m}$. Therefore, the selected indices depend on the choice of a so-called \emph{representative} sample for each orbit, which will be discussed later.
This notation emphasizes that selection acts on equivalence classes. 
This formulation allows us to define geometric coreset objectives (e.g., $k$-centers, coverage maximization) in the quotient space. Thus, the choice of quotient metric $d_G$ is crucial, and is the first step of our method, which is fully formalized and detailed in Section~\ref{subsec:group-invariant-coresets}. In short, the distance $d_G$ can be computed either by directly aligning orbits (which can be computationally prohibitive) or, in a more efficient and practical approach by using \textit{orbit-separating invariant functions}, i.e., learned $G$-invariant maps that can separate orbits and give an efficient way to realize the quotient geometry. This will be detailed in the following subsection.

The coreset $\mathcal{C}_{\calX/G}$ is then integrated into an invariant AL framework (a pipeline summarizing the proposed method is shown in Figure~\ref{fig:framework}).
At each iteration, after a batch is selected on the quotient space, one representative sample from each orbit is supplied to the oracle $\mathcal{O}$ for labeling, ensuring labeling budget is not wasted on multiple redundant copies within the same orbit.
Then, the classifier is re-trained with a properly designed orbit-averaged loss or invariant NN architecture, ensuring that selection, training, and prediction account for the data symmetries in a congruent way.
This is discussed in detail in Section~\ref{subsec:group-invariant-loss}.

\begin{figure}[htbp]
    \centering
    \includegraphics[width=1.0\linewidth]{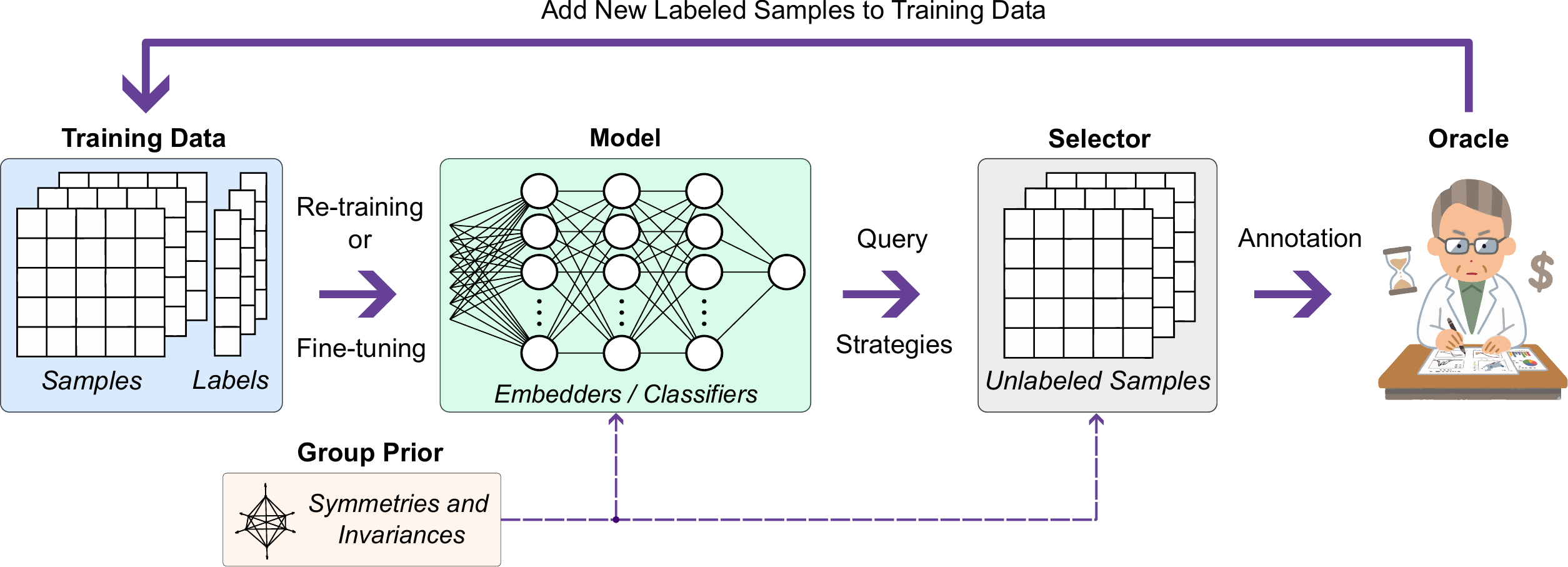}
    \caption{Overview of the \emph{group-invariant coreset} and AL pipeline.}
    \label{fig:framework}
\end{figure}

\subsection{Group-invariant Coresets}
\label{subsec:group-invariant-coresets}

\subsubsection{\textbf{Quotient space representation}}

Following the group-theoretic view of invariance, we therefore work on the quotient space $\mathcal{X}/G$ of the data under $G$ and formulate selection directly in the resulting quotient geometry. This way, orbits serve as natural summaries under invariance assumptions, which motivates operating on equivalence classes \cite{bloemreddy2020probabilistic,chen2020group} rather than raw samples for eliminating transformed duplicates in coreset construction.
The following definition makes the orbit-based equivalence precise.

\begin{definition}
\label{def:equivalence}
Given $\bx,\bx'\in\calX$, define the equivalence relation $\sim_G$ by writing $\bx\sim_G\bx'$ if there exists $g\in G$ such that $\bx=g\cdot \bx'$. 
\end{definition}

\begin{remark}
Definition~\ref{def:equivalence} is equivalent to the existence of $g,g'\in G$ such that $g\cdot\bx = g'\cdot\bx'$. If $g\cdot\bx = g'\cdot\bx'$, apply $g^{-1}$ on the left to obtain $\bx = (g^{-1}g')\cdot\bx'$, which matches~Definition~\ref{def:equivalence}.
\end{remark}

\begin{definition}
\label{def:orbit}
The orbit of a point $\bx\in\calX$ under the action of $G$ is the equivalence class of $\bx$ under the relation $\sim_G$, defined as $\mathscr{O}_G(\bx)=\{g \cdot \bx : g\in G\}$. The quotient $\mathcal{X}/G$ is defined as the set of orbits of $\calX$, that is, the equivalence classes $[\bx]_G$ modulo the equivalence relation $\sim_G$.
\end{definition}

\begin{figure}[t]
    \centering
    \includegraphics[width=1.0\linewidth]{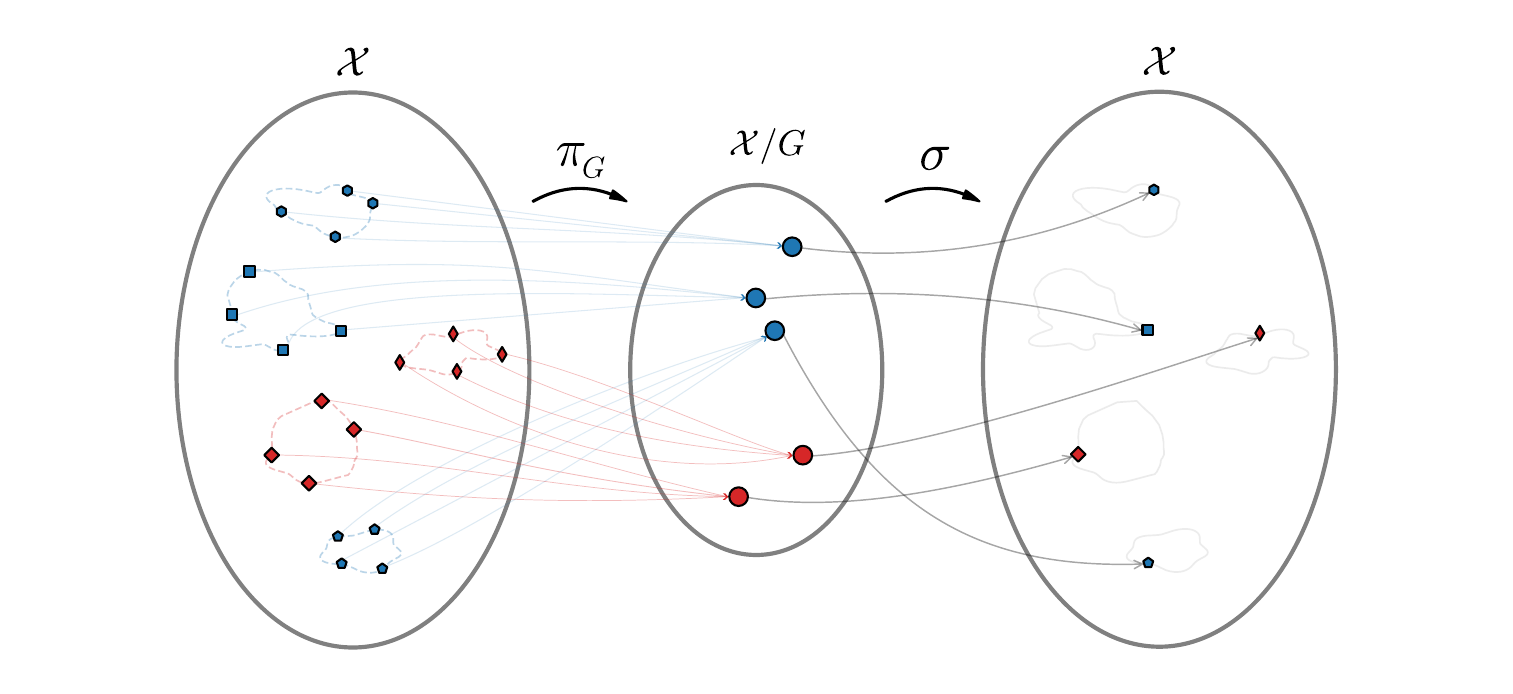}
    \caption{Visualizing the quotient mapping. Left: The input space $\mathcal{X}$ contains full orbits. Center: The projection $\pi_G$ maps each orbit to a single point in the quotient space $\mathcal{X}/G$. Right: The selector $\sigma$ maps each quotient point back to a unique canonical representative $\bar{\bx}$ in the input space.}
    \label{fig:mappings}
\end{figure}

The collection of orbits induced by the action of $G$ on $\calX$ yields a partition of $\calX$: each orbit $\mathscr{O}_G(\bx)$ is nonempty, any two distinct orbits are disjoint, and the union of all orbits equals $\calX$. Consequently the quotient set $\mathcal{X}/G$ is well defined, and there exists a projection mapping $\pi_G:\,\calX \to \mathcal{X}/G$, sending each point to its orbit/equivalence class. 
When a consistent choice of representative from each orbit is available (a canonical representative), the quotient may be identified with the set of those representatives; otherwise one works with orbits as described below.

\begin{remark}
For many choices of group actions (under additional regularity conditions) there may exist a canonical cross-section $\mathcal{S}\subseteq\mathcal{X}$ and a selector $\sigma:\mathcal{X}/G\to\mathcal{S}$ satisfying $\pi_G\big(\sigma(\mathscr{O}_G(\bx))\big)=\mathscr{O}_G(\bx)$ for every $\bx\in\mathcal{X}$. The canonical representative of $\bx$ is then $\bar{\bx}:=\sigma\big(\mathscr{O}_G(\bx)\big)\in\mathscr{O}_G(\bx)$. The composition $\sigma\circ\pi_G:\mathcal{X}\to\mathcal{S}$ maps each point to its chosen representative; it is constant on orbits and therefore generally not injective on $\mathcal{X}$ (injectivity holds only for the trivial action). 
\end{remark}

This is illustrated in Figure~\ref{fig:mappings}, where the left panel shows the raw input space $\mathcal{X}$, the center panel illustrates the quotient space $\mathcal{X}/G$ obtained via the projection $\pi_G$, and the right panel shows the selection of canonical representatives via the section~$\sigma$.
Note that once an orbit is included in the coreset during AL, the annotating oracle $\mathcal{O}$ only needs to label a single representative from that orbit, even if the pool contains multiple transformed copies (under the assumption that transformed samples share the same label--an assumption that will be slightly relaxed in a probabilistic setting later). A natural choice is to label the canonical representative $\bar{\bx}$ selected by $\sigma(\mathscr{O}_G(\bx))$. 

A key element of the framework is the definition of a suitable metric $d_G$ on the quotient space $\mathcal{X}/G$ that captures a relevant notion of similarity between orbits. This metric is essential for geometric coreset selection, and several constructions are possible.
In the following, we present two approaches: an idealized orbit-minimization distance, which is instructive but less practical to implement, and a more efficient one based on orbit-separating functions and invariant representations.

The following pseudo-distance inspired by \cite{bandeira2023estimationOrbitsInvariants} focuses on groups that act by isometries, for example, compact Lie groups with orthogonal representations.

\begin{definition}
\label{def:distance}
For $\bx,\bx'\in\calX$ and a group $G$ acting by isometries, define the ambient pseudo-distance
\begin{equation}\label{eq:dg-min}
\tilde d_G(\bx,\bx') \coloneqq \inf_{g\in G} \|\bx-g\cdot\bx'\|\ .
\end{equation}
\end{definition}

\begin{proposition}
If $G$ acts by isometries (i.e., $g\in G$ is an isometry w.r.t. the norm), then the ambient pseudo-distance can be written as
\begin{equation}
    \tilde d_G(\bx,\bx') \coloneqq \inf_{g,g'\in G} \|g\cdot\bx-g'\cdot\bx'\|\ .
\end{equation}
This pseudo-distance induces a metric on the quotient space $\mathcal{X}/G$ between orbits \cite{bandeira2023estimationOrbitsInvariants}, which we denote by $d_G\bigl(\mathscr{O}_G(\bx),\mathscr{O}_G(\bx')\bigr):=\tilde d_G(\bx,\bx')$. When $G$ is compact and the action is continuous, the infimum is attained.
\end{proposition}

\begin{proof}
    This result follows from \cite[Lemma~3.1]{hoyos2023diffusion}. Since $G$ acts by isometries, we have $\|\bx-\bx'\|=\|g\cdot\bx-g\cdot\bx'\|$. Applying $g^{-1}$ inside the norm implies the equivalence.
\end{proof}

Therefore, $\tilde d_G$ defines the metric between orbits in the quotient space $\mathcal{X}/G$, as required for our formulation. In the following, we will continue to use $d_G$ to denote the quotient metric between orbits, and $\tilde d_G$ for its equivalent ambient space pseudo-distance.
Examples of groups that act isometrically include the orthogonal group (acting by matrix multiplication) and the symmetric (permutation) group, although the latter is discrete rather than a Lie group.

\subsubsection{\textbf{Practical quotient metrics}}

Despite its simplicity, the ambient pseudo-distance in Definition~\ref{def:distance} is not applicable in general. When $G$ is non-compact, the infimum in \eqref{eq:dg-min} may not be attained. When $G$ is not isometric, $\tilde d_G$ is not independent of the choice of representative within the orbit, and thus does not define a quotient pseudo-distance.
Moreover, computing the solution to \eqref{eq:dg-min} can be computationally impractical for some groups (e.g., large permutation groups). 
While other canonical pseudo-distances can be defined on quotient spaces in full generality through minimizations over sequences \cite{burago2022courseMetricGeometryBook}, it is typically intractable to compute in practice. 
These issues motivate an alternative construction explained next. 

Our goal is to define a tractable distance on the quotient space $\mathcal{X}/G$ that is semantically meaningful and easy to compute. We thus define a quotient metric using the notion of \textit{orbit-separating invariants} \cite{nguyen2023learningSymmetrizationOrbitDistance,dym2024lowDimensionalInvariantEmbeddings}, i.e., functions $h:\calX\times\calX\to \mathbb{R}_+$ that (i) are $G$-invariant ($h(\bx)=h(g\cdot\bx)$ for all $g\in G$) and (ii) separate orbits:
\begin{equation} 
    \label{eq:orbit_separating_inv_i}
    h(\bx)\neq h(\bx') \quad\Longleftrightarrow\quad \mathscr{O}_G(\bx)\neq \mathscr{O}_G(\bx') \,.
\end{equation}

When such an $h$ is available, a practical quotient distance is obtained from the Euclidean distance in the image of $h$:
\begin{equation}\label{eq:d_G_psi}
    d_{G,h}(\bx,\bx') \coloneqq \|h(\bx)-h(\bx')\| \,.
\end{equation}
Various definitions of $d_{G,h}(\bx,\bx')$ are possible.
For simple cases, a closed-form canonicalizer can be used, that is, if every orbit $\mathscr{O}_G(\bx)$ has a unique canonical representative $\sigma(\bx)=\bar\bx$, the map $h(\bx)=\bar\bx$ is an orbit-separating invariant, being $G$-invariant by construction and satisfying the condition in \eqref{eq:orbit_separating_inv_i}.
For example, $h(\bx)=\bx/\|\bx\|$ is a canonicalizer for positive rescaling, thus $d_{G,h}(\bx,\bx') = \|h(\bx)-h(\bx')\|$. 

Another appealing option is to consider an orbit-separating function $h_\psi$ as a neural network with parameters $\psi$, learned from unlabeled data using invariant representation learning \cite{benton2020learning}.
In that case, $h_\psi$ is trained to be $G$-invariant while separating distinct orbits, and the induced distance provides a useful quotient metric for sample selection. 
Detailed constructions and learning-based methods for orbit-separating invariants (including polynomial and symmetrization approaches for groups such as the symmetric, orthogonal and general linear groups) are discussed in \cite{nguyen2023learningSymmetrizationOrbitDistance,dym2024lowDimensionalInvariantEmbeddings}.
Practical implementations can also leverage architectures with built-in invariance \cite{cohen2016group,weiler2024gauge}, group-invariant autoencoders \cite{winter2022unsupervised} or augmentation-based representation learning \cite{gidaris2018unsupervised,chen2020simclr,caron2021dino}.
Throughout the paper (including experiments) $d_{G,h}$ is the practical surrogate that replaces the ideal quotient metric $d_G$ and its equivalent ambient space pseudo-distance $\tilde d_G$.

\subsubsection{\textbf{Group-invariant coresets with quotient k-centers selection}}
\label{sec:group-invariant-k-centers}

Based on the quotient space representation presented in the previous subsection, we can formulate the coreset selection directly on the quotient space $\mathcal{X}/G$ to eliminate redundancy by using a geometric criterion.
While many coreset objectives are possible, we focus here on the $k$-center objective \cite{Gonzalez1985,sener2018coresets} due to its simplicity and popularity in AL. The proposed framework can also be extended to other geometric criteria (e.g., geometric medians, probabilistic sampling or other coverage maximization criteria), the key difference being the use of the quotient metric $d_G$ which ensures the selection is performed among equivalence classes.

Let $\mathcal{U}=\{1,\ldots,N_U\}$ denote the index set of the unlabeled pool $\mathcal{D}_U$. We formulate the $G$-invariant $k$-centers-based GRINCO coreset selection as follows.

\begin{definition}[$k$-centers GRINCO]
A $G$-invariant $k$-center coreset is defined as the set of orbits that solve
\begin{equation}\label{eq:kcovers-quotient}
    \calC_{\rm idx} \,\in\,\, \mathop{\arg\min}_{\substack{S\subseteq \mathcal{U}\\ |S|=K}}\;
    \max_{n\in\mathcal{U}}\;
    \min_{s\in S}\; d_G\bigl(\mathscr{O}_G(\bx_n),\mathscr{O}_G(\bx_s)\bigr).
\end{equation}
This corresponds to the $k$-center objective in the quotient space $\mathcal{X}/G$.
\end{definition}

Here $\calC_{\rm idx}$ denotes the indices of the samples that serve as representatives for the selected orbits in the unlabeled pool, and the coreset itself is given by $\calC_{\calX/G} = \{\mathscr{O}_G(\bx_s) : s\in \calC_{\rm idx}\}$. Note that problem \eqref{eq:kcovers-quotient} might have multiple optimal solutions corresponding to different choices of representatives for the same set of orbits.

\begin{remark}
The $k$-centers GRINCO objective in \eqref{eq:kcovers-quotient} can also be written equivalently using the ambient pseudo-distance $\tilde d_G$ in Definition~\ref{def:distance}, or the orbit-separating invariant-based distance $d_{G,h}$ in \eqref{eq:d_G_psi}. Since these choices of distance all define a quotient geometry, they all lead to orbit selection criteria.
\end{remark}

To solve \eqref{eq:kcovers-quotient}, a practical approach is to adopt a greedy $k$-center algorithm (farthest-first traversal) using one of the aforementioned quotient metrics, such as $d_{G,h}$ (or $\tilde{d}_G$ if tractable) as the distance function. This procedure, which we refer to as \emph{Orbit-$k$-center}, iteratively selects the orbit maximizing the minimum quotient distance to the current coreset.

Importantly, downstream machine learning methods (classifier, regressor or clustering routine) should be made $G$-invariant as well to fully benefit from an invariant coreset. This will be discussed in the context of AL in the following. Otherwise, the model cannot distinguish which discarded samples were redundant under $G$.

\subsection{Efficient AL with group invariant coresets}\label{subsec:group-invariant-loss}

\subsubsection{Active learning with GRINCO}\label{subsec:slim-learning}

Let us denote the full dataset by $\mathcal{D}_{\mathrm{full}}=\{(\bx_n,y_n)\}_{n=1}^{N}$, containing the labels $y_n$ for all samples.
The action of $G$ partitions the observed indices into $\{\mathcal{I}_i\}_{i=1}^M$, $\mathcal{I}_i\subset \{1,\ldots,N\}$, where $M$ is the number of distinct orbits in the dataset. In particular, $\mathscr{O}_i$ denotes the orbit containing the samples indexed in $\mathcal{I}_i$ such that $\bx_n$ and $\bx_{n'}$ lie in the same orbit $\mathscr{O}_i$ if and only if $n,n'\in\mathcal{I}_i$. We also denote by $\bar{\bx}_i := \sigma(\mathscr{O}_i)$ the canonical representative of orbit $\mathscr{O}_i$.
In this case, $\mathcal{D}_{\mathrm{full}}$ can be partitioned over orbits as:
\begin{equation} \label{eq:orbit-dataset-partitioning}
\mathcal{D}_{\mathrm{full}} \;=\; \bigsqcup_{i=1}^{M} \big\{ (\bx_n,\, y_n) \big\}_{n\in\mathcal{I}_i} .
\end{equation}

We can define a measure of the orbit size in the dataset by introducing $\alpha_i := |\mathcal{I}_i|/\sum_{j=1}^{M} |\mathcal{I}_j|$ which represents the proportion of samples in orbit $\mathscr{O}_i$, with $\sum_{j=1}^{M} \alpha_j=1$. If we retain one representative for each measured orbit and attach to it a queried label $y_i$, we obtain the \emph{orbit summarized} dataset
\begin{equation} \label{eq:slim-dataset-summarization}
    \mathcal{D}_{\mathrm{orbit}} \;=\; \big\{\, (\bar{\bx}_i,\, y_i,\, \alpha_i) \,\big\}_{i=1}^{M} \,.
\end{equation}

Based on a set of representative samples $\{\bar{\bx}_i\}_{i=1}^K$ selected by GRINCO\footnote{With a slight abuse of notation, we use the same notation to index the representatives $\bar{\bx}_i$ in the orbit summarized dataset $\mathcal{D}_{\mathrm{orbit}}$ and in the coreset.} according to~\eqref{eq:kcovers-quotient}, we also consider a more general \emph{weighted representative coreset} 
\begin{align}
    \mathcal{C} = \big\{\, (\bar{\bx}_i,y_i,w_i) \, \big\}_{i=1}^{K} \,,
\end{align}
where $K$ is the coreset size and $w_i>0$ is the weight of its $i$-th element. Different approaches can be used to define $w_i$; we compute it as $w_i=\textsf{nno}(\bar{\bx}_i)/N$, where $\textsf{nno}(\bar{\bx}_i)$ denotes the number of samples in $\mathcal{D}_{\mathrm{full}}$ whose closest orbit is $[\bar{\bx}_i]_G$. This gives a measure of the contribution of $[\bar{\bx}_i]_G$ in the full dataset.
Note that in the special case $K=M$ in which one representative is retained for every measured orbit, we have that $w_i=\alpha_i=|\mathcal{I}_i|/N$.
It is instructive to first consider the ideal case of deterministic invariance, where all transformed versions of a given input share the same label. In that case, training could be performed on the orbit summarized dataset with one label per orbit without changing the empirical objective, as illustrated in the following.
\begin{example}[Deterministic invariance]\label{ex:deterministic-invariance}
Suppose that all samples in each orbit share the same label, that is, $y_n = y_{n'}$ whenever $n,n'\in\mathcal{I}_i$. If $f_{\bw}$ is $G$-invariant by construction, i.e., $f_{\bw}(g \cdot \bx)=f_{\bw}(\bx)$ for all $g\in G$, then
\begin{align}
    \frac{1}{N}\sum_{(\bx_n,y_n)\in \mathcal{D}_{\mathrm{full}}} \mathcal{L}\big(f_{\bw}(\bx_n),\, y_n\big)
    &= \frac{1}{N}\sum_{i=1}^{M}\;\sum_{n\in\mathcal{I}_i} \mathcal{L}\big(f_{\bw}(\bx_n),\, y_i\big) 
    \nonumber \\
    & \hspace{-4.5ex}= \sum_{(\bar\bx_i,y_i,\alpha_i)\in \mathcal{D}_{\mathrm{orbit}}} \alpha_i \,\mathcal{L}\big(f_{\bw}(\bar{\bx}_i),\, y_i\big)
    \nonumber
\end{align}
where we used the fact that $f_{\bw}(\bx_n)=f_{\bw}(\bar{\bx}_i)$ for all $n\in\mathcal{I}_i$.
This shows that minimizing the empirical risk over the full dataset $\mathcal{D}_{\mathrm{full}}$ is equivalent to minimizing a weighted risk over the canonical representatives contained in the orbit summarized dataset $\mathcal{D}_{\mathrm{orbit}}$ with a single representative per orbit.
\end{example}

Note that Example~\ref{ex:deterministic-invariance} assumes that all samples in each orbit share the same label, which may not hold in practice due to label noise or imperfect invariance. To address this limitation, we will consider a more general probabilistic invariance framework based on Definition~\ref{def:label-invariance}, leveraging orbit-averaged losses (as in \eqref{eq:orbit-averaged-loss}) to accommodate stochastic labels within orbits while still enabling efficient learning with weighted representative coresets produced by GRINCO.

\subsubsection{A statistical invariance-aware learning objective}
We now cast the invariance-aware objective in a statistical framework. Let $\bbQ$ denote a probability measure on $G$, which will be used to promote invariance during training (e.g., by using data augmentation with random transformations sampled from $\bbQ$~\cite{chen2020group}), in such a way as to promote $G$-invariance in both the quotient-space coreset selection metric $d_{G,h}$ and the classification model $f_{\bw}$\footnote{When available, a classifier with an architecture that is group invariant by design can be equivalently used to the same effect.}. 
Coupling the $G$-invariance assumptions used in coreset selection and in AL model training is important to achieve the desired performance.

We define the orbit-averaged population risk as the expectation (with respect to $p(\bX,Y)$) of the per-sample orbit-averaged loss $\mathcal{L}^{\bbQ}(\bx,y,f_{\bw})$ as
\begin{align}\label{eq:loss-orb-full}
    \mathscr{R}^{\bbQ}(f_{\bw})
    &\,\coloneqq\, \bbE_{p(\bX,Y)}\bigl[\mathcal{L}^{\bbQ}(\bX,Y,f_{\bw})\bigr]
    \nonumber \\
    &\,\,\,=\, \bbE_{p(\bX,Y)}\, \bbE_{\bbQ(g)}\, \bigl[\mathcal{L}\bigl( f_{\bw}(g\cdot \bX),\, Y \bigr)\bigr]\,.
\end{align}

\begin{remark}
If $f_{\bw}$ is $G$-invariant by construction, the orbit-averaged loss reduces to the standard loss: $\mathcal{L}^{\bbQ}(\bx,y,f_{\bw}) = \mathcal{L}(f_{\bw}(\bx),y)$, and so does the risk, $\mathscr{R}^{\bbQ}(f_{\bw})=\mathscr{R}(f_{\bw})$.
\end{remark}

The outer expectation in the orbit-averaged risk in \eqref{eq:loss-orb-full} can be approximated empirically using the samples in the weighted coreset $\mathcal{C}$ as follows:
\begin{equation} \label{eq:orbit-averaged-risk-coreset}
    \hat{\mathscr{R}}^{\bbQ}_{\mathcal{C}}(f_{\bw})
    \,\coloneqq\, \sum_{i=1}^{K} w_i\, \bbE_{\bbQ(g)} \mathcal{L}\bigl(f_{\bw}(g\cdot \bar{\bx}_i), y_i\bigr).
\end{equation}
The orbit-averaged loss in \eqref{eq:orbit-averaged-risk-coreset} along with the weighted representative coreset will then be used in a batch AL iteration pipeline, which is described in more detail in Section~\ref{subsec:invar-batch-al}.

In practice, the minimization of $\hat{\mathscr{R}}^{\bbQ}_{\mathcal{C}}(f_{\bw})$ can be performed with a Monte Carlo approximation for the expectation $\bbE_{\bbQ(g)}$, using data augmentation with random transforms drawn from $\bbQ$ during training \cite{chen2020group}. Other approaches include the use of constrained optimization with penalties \cite{manolache2025learning} that promote invariance.

Note that for the empirical orbit-averaged risk \eqref{eq:orbit-averaged-risk-coreset}, equivalence between training on the full dataset and training on the weighted coreset does not hold as in Example~\ref{ex:deterministic-invariance} due to stochasticity in the labels or inexact invariance. 
The next subsection addresses this difficulty through a theoretical analysis that formalizes how coverage in a quotient metric controls the generalization error when training on the weighted coreset~$\mathcal{C}$.


\subsection{Generalization}
\label{subsec:generalization}

We now analyze the excess risk incurred by selecting a group-invariant weighted representative coreset $\mathcal{C}$ with coverage radius $\varepsilon$ and training with the orbit-averaged loss under the probabilistic invariance from~\eqref{eq:label-invariance-eq}, compared to the empirical risk that would be obtained by training on the full dataset, which is defined by
\begin{align} \label{eq:orbit-averaged-risk-full}
    \hat{\mathscr{R}}_{\mathrm{full}}^{\bbQ}(f_{\bw}) & =\frac{1}{N}\sum_{n=1}^N \bbE_{g\sim\bbQ}\mathcal{L}\big(f_{\bw}(g\cdot\bx_n),y_n\big) \,.
\end{align}
The result separates the generalization error into concentration effects, and coreset-specific approximation errors driven by label uncertainty and coverage in the quotient space under $d_{G}$. It is formalized in the following theorem.

\begin{theorem}[$G$-invariant coreset generalization] \label{thm:generalization_gcoresets}
Let $\mathcal{D}_{\mathrm{full}}=\{(\bx_n,y_n)\}_{n=1}^N$ be the full dataset and $\mathcal{C}=\{(\bar{\bx}_i,y_i,w_i)\}_{i=1}^K$ a group-invariant weighted representative coreset. Consider the orbit-averaged empirical risk $\hat{\mathscr{R}}_{\mathcal{C}}^{\bbQ}$ defined analogously to \eqref{eq:orbit-averaged-risk-coreset}, and suppose that:
\begin{itemize}
    \item[A1:] The conditional label distribution is $G$-invariant, i.e., $p(Y|\bX)=p(Y|g\cdot \bX)$ for all $\bX$. 
    \item[A2:] The loss is bounded, i.e., $\sup_{\boldsymbol{p}\in\Delta_C,\, y\in\mathcal{Y}} \mathcal{L}(\boldsymbol{p},y) \le L_{\max}$, and both the labeling function and loss function are Lipschitz continuous with respect to the quotient pseudo-metric $d_{G,h}$, that is, there exist constants $L_{p}, \, L_{\mathcal{L}}$ such that for any $c\in\mathcal{Y}$ and $\bx,\bx'\in\mathcal{X}$:
    \hspace*{-\leftmargin}%
    \makebox[\linewidth]{%
    \begin{minipage}{\linewidth}
    \begin{align}
        & |p(Y=c|\bx) - p(Y=c|\bx')|  \leq L_{p} d_{G,h}(\bx,\bx')
        \label{eq:thm_a_lip_p}
        \\
        & \big|\bbE_{\bbQ(g)}\big\{\mathcal{L}(f_{\bw}(g\cdot \bx),c) - \mathcal{L}(f_{\bw}(g\cdot \bx'),c)\big\} \big|  \leq L_{\mathcal{L}} d_{G,h}(\bx,\bx')
        \label{eq:thm_a_lip_L}
    \end{align}
    \end{minipage}%
    }
    
    \item[A3:] The representative coreset $\calC$ is an $\varepsilon$-cover of $\mathcal{D}_{\mathrm{full}}$ in the quotient metric, i.e., $\max_{n\in[N]} \min_{i\in[K]} d_{G,h}(\bx_n,\bar{\bx}_i)\le\varepsilon.$

\end{itemize}
Then, with probability at least $1-\gamma$ over the draw of $\mathcal{D}_{\mathrm{full}}$, the generalization error of the representative coreset $\mathcal{C}$ satisfies
\begin{align}
    & \hspace{-5ex} \underbrace{\big|\hat{\mathscr{R}}_{\mathcal{C}}^{\bbQ}(f_{\bw}) - \mathscr{R}(f_{\bw})\big|}_{\text{coreset generalization error}}
    \,\,\leq\,\, \underbrace{\big|\hat{\mathscr{R}}_{\mathrm{full}}^{\bbQ}(f_{\bw}) - \mathscr{R}(f_{\bw})\big|}_{\text{full dataset generalization term}} \nonumber \\
    & \hspace{7ex}  + \underbrace{\sqrt{\frac{2 \ln(2/\gamma)}{N} V(\bbQ) }}_{\text{I}} + \underbrace{\frac{2\,L_{\max}\,\ln(2/\gamma)}{3N}}_{\text{II}} 
    \nonumber \\
    & \hspace{7ex} + \underbrace{2 L_{\max} \sum_{i=1}^{K} w_i \eta_i}_{\text{III}} + \underbrace{(L_{\mathcal{L}}  + C L_{p} L_{\max}) \varepsilon}_{\text{IV}} 
    \label{eq:thm_generalization_bound}
\end{align}
where $\mathscr{R}(f_{\bw}) =\bbE_{p(\bX,Y)}\big[\mathcal{L}(f_{\bw}(\bX), Y)\big]$ is the population risk, $\eta_i= 1-p(Y=y_i| \bar{\bx}_i)$ is the labeling uncertainty in the $i$-th representative, and $V(\bbQ)$ is the variance-related term:
\begin{align}
    V(\bbQ) {}={} & \bbE_{p(\bX)}\{\bbE_{\bbQ(g)}\{ \widetilde{\xi}(g\cdot \bX) \}\} + \mathrm{Var}_{p(\bX)}(\widetilde{\zeta}(\bX)) 
    \nonumber\\
    & - \bbE_{p(\bX)} \{\mathrm{Var}_{\bbQ(g)}(\widetilde{\zeta}(g\cdot \bX))\} \,.
    \label{eq:thm_variance_term}
\end{align}
\end{theorem}

The proof of Theorem~\ref{thm:generalization_gcoresets} appears in Appendix~\ref{app:generalization-proof}, where the auxiliary quantities $\widetilde{\xi}$ and $\widetilde{\zeta}$ are also defined.

The bound in \eqref{eq:thm_generalization_bound} clarifies the behavior of the proposed AL strategy by separating finite-sample effects from coreset approximation effects. The first term on the right-hand side is the standard generalization gap for the full dataset under the orbit-averaged objective. It measures the deviation between the population risk and the empirical orbit-averaged loss when all $N$ samples are used, so it is independent of coreset selection.
Terms I and II quantify the concentration of the empirical full-dataset loss around its conditional expected-label counterpart that appears in the proof. 
Term I is the dominant variance-driven component that scales like $O(1/\sqrt{N})$. The variance constant $V(\bbQ)$ aggregates within-orbit and across-orbit variability through a law-of-total-variance decomposition and includes a negative correction that reduces variance when averaging over $\bbQ$ (in general, the larger the variance of $\bbQ$, the more significant the reduction will be). 
Term II is a boundedness correction from Bernstein's inequality~\cite{maurer2019bernsteinBoundedInteraction,sridharan2002gentleConcentrationIneq} and scales like $O(1/N)$, thus this term tends to be comparatively very small for practical values of $N$.

The third and fourth terms (III and IV) capture coreset-specific effects. Term III measures the expected label uncertainty at the selected representatives, weighted by the coreset weights $w_i$. It becomes zero when the oracle provides exact (deterministic) labels, and it can be large when the labeling uncertainty is significant and the size of the coreset $K$ is small. 
Term IV is the approximation error induced by approximating the loss over the coreset. It scales with the covering radius $\varepsilon$ and with the Lipschitz constants of the labeling function and the orbit-averaged loss. Coverage is measured in the quotient space $\mathcal{X}/G$ using the orbit-separating pseudo-distance $d_{G,h}$, so redundant within-orbit variations do not inflate the radius.

\subsection{Active learning pipeline with GRINCO}
\label{subsec:invar-batch-al}

We complete the proposed group-invariant coreset AL framework by integrating it into a pool-based loop: at each round, we compute a quotient-space coreset, query one representative per orbit, and retrain with an invariant or orbit-averaged loss (Algorithm~\ref{alg:grinco-light}). The acquisition rule $\mathcal{A}_{G}$ selects orbits in $\mathcal{D}_{U}$ by quotient-space coverage relative to the current labeled set $\mathcal{D}_{L}$, represented by the weighted representative coreset $\mathcal{C}$ (Section~\ref{subsec:group-invariant-loss}), using $\tilde{d}_G$ or $d_{G,h}$. In the orbit-$k$-centers instantiation, this is implemented by greedy farthest-first selection in the quotient space.

\begin{algorithm}[htbp]
\caption{GRINCO-Based AL Pipeline}
\label{alg:grinco-light}
\begin{algorithmic}[1]
\REQUIRE Unlabeled pool $\mathcal{D}_{U}$, initial labeled set $\mathcal{D}_{L}$, group $G$, group measure $\bbQ$, acquisition rule $\mathcal{A}_{G}$,
quotient distance $\tilde{d}_G$ or $d_{G,h}$, batch size $b$, label budget $B$.
\STATE Initialize classifier $f_{\bw}$ on $\mathcal{D}_{L}$ and decompose $\mathcal{D}_{L}$ as described in \eqref{eq:orbit-dataset-partitioning}--\eqref{eq:slim-dataset-summarization} to obtain the initial weighted representative coreset $\mathcal{C}$
\WHILE{budget is not exhausted and classification accuracy is not met}
\STATE Select $b$ new orbits $\{\mathscr{O}_i\}_{i=1}^b$ from $\mathcal{D}_{U}$ based on the GRINCO criterion in Section~\ref{sec:group-invariant-k-centers} (e.g., top-$b$ farthest-first traversal in the quotient space)
\STATE For each selected orbit $\mathscr{O}_i$, choose a representative $\bar{\bx}_i\in\mathcal{D}_{U}$ and query the oracle $\mathcal{O}$ for its label $y_i$
\STATE For every sample $\bx_n$ in $\mathcal{D}_{U}$ such that $\bx_n$ belongs to a selected orbit $\mathscr{O}_i$, remove $\bx_n$ from $\mathcal{D}_{U}$ and add $(\bx_n,y_i)$ to $\mathcal{D}_{L}$
\STATE Compute the weights $w_i$ as described in Section~\ref{subsec:group-invariant-loss} and update the weighted representative coreset $\mathcal{C}\leftarrow \mathcal{C} \bigcup\{(\bar{\bx}_i,y_i,w_i)\}_{i=1}^b$
\STATE Retrain the model $f_{\bw}$ on $\mathcal{C}$ using an invariant or orbit-averaged loss such as \eqref{eq:orbit-averaged-risk-coreset}
\ENDWHILE
\STATE \textbf{return} $f_{\bw}$ and the weighted representative coreset $\mathcal{C}$.
\end{algorithmic}
\end{algorithm}

\section{Experiments}\label{sec:experiments}

We evaluate the proposed group-invariant coreset framework in pool-based AL. Each experiment starts from an initial labeled set and a large unlabeled pool, then alternates training and batch acquisition until reaching a fixed label budget $B$, as in Algorithm~\ref{alg:grinco-light}. We first study a controlled scale-invariant synthetic dataset (Section~\ref{sec:exp:rays}), then image benchmarks with rotation invariance on CIFAR-10, STL-10, and MNIST (Section~\ref{sec:exp:image-benchmarks}).
We compare GRINCO with baseline methods in terms of classification accuracy and orbit efficiency on the selected coreset. All experiments were implemented in PyTorch and run on an NVIDIA RTX A4500 GPU with 16GB VRAM.

\subsection{Synthetic 2D scale-invariant rays dataset} 
\label{sec:exp:rays}

{We first consider a synthetic dataset with positive-rescaling invariance, so that each orbit is a ray from the origin. Labels depend only on the ray direction, not on the radius. This illustrative experiment compares random sampling, a Euclidean $k$-centers baseline, and GRINCO, showing reduced redundancy and improved label efficiency.}

\subsubsection{\textbf{Positive rescalings invariance model}}
\label{sec:rays:group}

Let $\mathcal{X}:=\mathbb{R}^2\setminus\{\boldsymbol{0}\}$ and $G:=\mathbb{R}_{+}$ act by positive rescaling, $g\cdot\bx=g\bx$ for $g\in\mathbb{R}_{+}$. The orbit of $\bx$ is the ray $\mathscr{O}_G(\bx)=\{g\cdot\bx:g\in\mathbb{R}_{+}\}$. A canonical orbit-separating map is the unit-norm projection
\begin{equation}
  h(\bx)\;:=\;\frac{\bx}{\|\bx\|_2}\in\mathbb{S}^1,
  \label{eq:rays:canonicalizer}
\end{equation}
where $\mathbb{S}^1:=\{\bx\in\mathbb{R}^2:\|\bx\|_2=1\}$ and $h(g\cdot\bx)=h(\bx)$ for all $g\in\mathbb{R}_{+}$. The quotient distance is $d_{G,h}(\bx,\bx')=\|h(\bx)-h(\bx')\|_2$, which is zero when $\bx$ and $\bx'$ lie on the same ray.

We sample $M$ distinct rays by drawing angles $\{\theta_i\}_{i=1}^M\subset[0,2\pi)$ and defining $\bu_i:=(\cos\theta_i,\sin\theta_i)\in\mathbb{S}^1$, with orbit $\mathscr{O}_i=\{r\bu_i:r\in\mathbb{R}_{+}\}$. For each orbit $i$, radii $\{r_{ij}\}_{j=1}^{\mu_i}$ are drawn i.i.d. from a log-uniform law, i.e., $\log r_{ij}\sim\mathrm{Unif}(\log r_{\min},\log r_{\max})$, and samples are $\bx_{ij}:=r_{ij}\bu_i$. The unlabeled pool is $\mathcal{D}_U=\{\{\bx_{ij}\}_{j=1}^{\mu_i}\}_{i=1}^M$ with $N:=|\mathcal{D}_U|=\sum_{i=1}^M \mu_i$, so large $\mu_i$ induce orbit redundancy. Labels are chosen to be scale-invariant (i.e., $y(g\cdot\bx)=y(\bx)$), and depend only on the orbit index (direction), not on radius:
\begin{equation}
  y_{ij}:=y_i,\qquad \forall j\in\{1,\dots,\mu_i\}.
  \label{eq:rays:labels}
\end{equation}
For a $4$-class demonstration, we use $M=C=4$ with rays separated by $90^\circ$ as $\theta_i=\pi/4+(i-1)\pi/2$, and $y_i=i$.

Note that in this setting, $\mathcal{D}_U$ can be partitioned into orbits indexed by $\mathcal{I}_i=\{(i,j):j=1,\dots,\mu_i\}$ as in \eqref{eq:orbit-dataset-partitioning}, and each orbit can be summarized by a single representative $\bar{\bx}_i=\bu_i=h(\bx_{ij})$ with label $y_i$ and weight $w_i=\mu_i/N$ as in \eqref{eq:slim-dataset-summarization}. 
Since the labels are scale-invariant, any invariant classifier $f_{\bw}(g\cdot\bx)=f_{\bw}(\bx)$ satisfies the identity $\hat{\mathscr{R}}_{\mathrm{full}}(f_{\bw})=\hat{\mathscr{R}}_{\mathcal{C}}(f_{\bw})$ as in Example~\ref{ex:deterministic-invariance}. Thus, one single representative per ray is sufficient to recover the same empirical risk as the full dataset under the GRINCO framework.

\subsubsection{\textbf{Orbit coverage and redundancy evaluation}}
\label{sec:rays:random}

We can quantify the expected orbit redundancy under a random coreset acquisition. Let $p_i:=\mu_i/N$ be the probability that a uniformly sampled point from $\mathcal{D}_U$ lies in orbit $i$, and let $U_B$ be the number of distinct orbits observed in the coreset after $B$ random queries. 
Then, assuming a sampling with replacement approximation for large $N$, the expected number of redundant queries (i.e., queries that do not introduce a new orbit) is
\begin{equation*}
  \mathbb{E}[\mathrm{Redundancy}]:=B-\mathbb{E}[U_B]=B-\sum_{i=1}^M \Bigl(1-(1-p_i)^B\Bigr).
\end{equation*}
Assuming that the orbits are balanced ($\mu_i=\mu$ for all $i$, and $p_i=1/M$) gives
\begin{equation*}
  \mathbb{E}[\mathrm{Redundancy}]
  =
  B-M\left(1-\left(1-\tfrac{1}{M}\right)^B\right).
\end{equation*}
Thus, in the case $M=B=4$, $\Pr(\text{no redundancy})\approx0.094$, so over $90\%$ of random batches contain at least one redundant sample. For unbalanced $\{\mu_i\}$, large orbits dominate random draws and redundancy is typically even higher.

\subsubsection{\textbf{Experimental configuration}}
\label{sec:rays:methods}

At budget $B$, we compare three coreset rules. (a) \emph{Random} uniformly samples $B$ points from $\mathcal{D}_U$, often yielding multiple samples from the same orbit. (b)~The \emph{Euclidean} baseline applies farthest-first $k$-center selection in the raw input space $\mathcal{X}$ by the update $S\leftarrow S\cup\{\bx^\star\}$, with $\bx^\star$ being the farthest point from the existing coreset $S$ at each iteration.
(c) GRINCO runs greedy $k$-center in the quotient space $\mathcal{X}/G$ using the pseudo-distance $d_{G,h}(\bx,\bx')=\|h(\bx)-h(\bx')\|_2$ with $h$ from \eqref{eq:rays:canonicalizer}. It iteratively selects a point $\bx^\star$ that maximizes the smallest quotient distance (which is invariant along rays) to the existing coreset $S$ as $\bx^\star \in \arg\max_{\bx\in \mathcal{D}_U}\;\min_{\bx'\in S} d_{G,h}(\bx,\bx')$.

We implement this synthetic experiment with $M=C=4$ rays (classes), with scale-invariant labels $y_i=i$ and orbit sizes $(\mu_1,\mu_2,\mu_3,\mu_4)=(400,200,100,100)$, yielding a pool of $N=800$ unlabeled samples. Along each ray, radii are sampled log-uniformly with $(r_{\min},r_{\max})=(0.1,10)$.
We evaluate the coreset redundancy for different acquisition budgets $B \in \{1,2,3,4,5,6,8,10\}$.
Classification performance is evaluated with a $1$-NN classifier (using the Euclidean metric in $\mathcal{X}$ for the \emph{Uniform} and \emph{Euclidean} baselines, and the quotient distance $d_{G,h}$ for GRINCO).
In addition to the overall test accuracy (OA), we report the \emph{orbit efficiency}, defined as the ratio of the number of distinct orbits represented in the coreset $U_B$ to the total number of selected samples $B$:
\begin{equation}
  \mathrm{\eta}(B) \;:=\; \frac{U_B}{B} \in (0,1] \,.
\end{equation}

\subsubsection{\textbf{Results and discussion}}
\label{sec:rays:results}

Figure~\ref{fig:rays:selection:B4} shows a sample of coreset acquisition at $B=4$ and the unlabeled pool.
The \textit{random} acquisition often selects multiple points on the same ray while missing some rays entirely.
The \emph{Euclidean coreset} baseline improves coverage, but it can still select multiple samples from the same orbit (leading to coreset redundancy). Samples along a ray at different radii might appear to be far apart in $\mathbb{R}^2$.
In contrast, the proposed GRINCO (operating in the quotient space with the distance $d_{G,h}$) selects at most one sample per orbit when $B$ does not exceed the number of distinct orbits $M$.

\begin{figure}[t]
    \centering
    \includegraphics[width=1.0\linewidth]{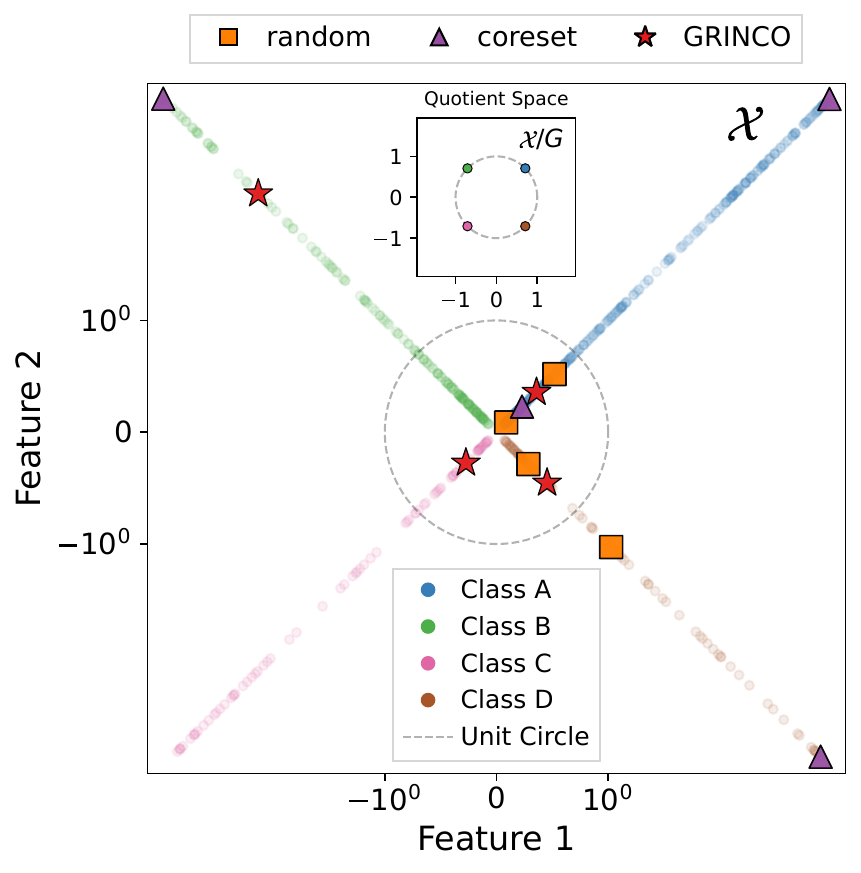} \\
    %
    \caption{Selected coresets for the \emph{rays dataset} at budget $B=4$ using \emph{random}, \emph{Euclidean coreset}, and GRINCO. Points are colored by class, and marker type indicates the selected samples. Class-wise counts are \textit{random} $(2,0,0,2)$, \textit{Euclidean coreset} $(2,1,0,1)$, and GRINCO $(1,1,1,1)$ for classes $(A,B,C,D)$.}
    \label{fig:rays:selection:B4}
\end{figure}

Figure~\ref{fig:rays:orbit-efficiency} shows orbit efficiency $\eta(B)$ as a function of budget, averaged over $30$ Monte Carlo runs. As expected, GRINCO achieves $\mathrm{\eta}(B)=1$ for $B\leq 4$, with one representative per orbit, while \textit{random} yields redundant selections even at very low budgets.
The \textit{Euclidean coreset} baseline achieves intermediate performance at $B=3,4$.
For $B>4$, orbit efficiency decreases for all methods as the budget exceeds the number of distinct orbits $M=4$.

\begin{figure}[htbp]
    \centering
    \includegraphics[width=0.8\linewidth]{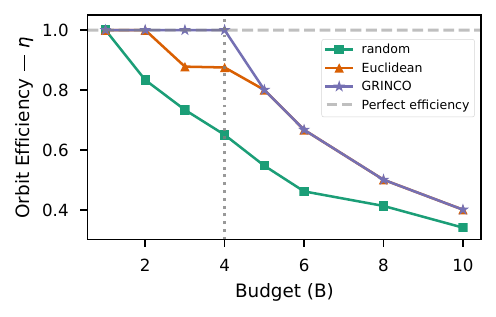} \\
    \caption{Orbit efficiency $\mathrm{\eta}$ on the \emph{rays dataset} (mean over $30$ runs) for \emph{random} sampling, \emph{Euclidean coreset}, and GRINCO. GRINCO achieves $\mathrm{\eta}=1$ for $B\leq M=4$.}
    \label{fig:rays:orbit-efficiency}
\end{figure}

The improvement in orbit efficiency translates directly into higher classification accuracy at low budgets. At $B=4$, \textit{random} attains $\mathrm{\eta}=0.65 \pm 0.15$ with $79.2\% \pm 9.9\%$ accuracy, while \textit{Euclidean coreset} improves to $\mathrm{\eta}=0.88 \pm 0.13$ and $93.8\% \pm 6.2\%$ accuracy (mean $\pm$ std). 
GRINCO achieves perfect orbit efficiency $\mathrm{\eta}=1.0 \pm 0.0$ and $100.0\% \pm 0.0\%$ accuracy, which is consistent with the orbit-invariant labeling model. Figure~\ref{fig:rays:selection:B4} explains this behavior: \textit{random} can spend the budget on a subset of the rays (especially on orbits that contain more samples in the unbalanced pool), reducing $U_B$ and leaving classes unlabeled; \emph{Euclidean coreset} spreads points in ambient space but can still pick points at multiple radii on the same ray; GRINCO enforces orbit diversity because $d_{G,h}$ is invariant to changes along the same ray, which guarantees the selection of at most one point per orbit when $B\leq M$.
This behavior persists across budgets.
GRINCO maintains perfect orbit efficiency up to $B=4$; for $B\geq 5$ both GRINCO and \textit{Euclidean coreset} consistently achieved full orbit coverage ($U_B=4$), leading to $\mathrm{\eta}(B)=4/B$.
\textit{Random} remained below this ceiling even at $B=10$, showing sustained redundancy, which is exacerbated by imbalanced orbits.
This experiment shows that quotient-space orbit-aware selection (GRINCO) avoids redundancy and achieves maximal orbit coverage per label, unlike \emph{random} and \emph{Euclidean} baselines.

\subsection{Image classification with rotation invariance}
\label{sec:exp:image-benchmarks}

{We evaluate GRINCO on image datasets to assess the impact of quotient-space coreset selection on label efficiency in AL under rotation invariance. This section defines the dataset construction, quotient-space distance, classifier, and AL protocol, and then compares GRINCO with several baselines in terms of classification accuracy and orbit efficiency.}

\subsubsection{\textbf{Experimental setup}}
\label{sec:exp-setup}

\paragraph{\textit{Datasets}} We consider the CIFAR-10~\cite{krizhevsky2009learning}, STL-10~\cite{coates2011analysis} and MNIST~\cite{lecun2002gradient} datasets and their \emph{rotated} variants, which are constructed by adding explicit rotation-induced symmetry. We use a standard train/test split of the datasets. For each dataset we select a class-balanced set of source images $\{\tilde{\bx}_i\}_{i=1}^{N_{\text{src}}}$ from the training split (which will be the \emph{raw} dataset). Then, for the \emph{rotated} dataset variants (i.e., rotated CIFAR-10, rotated STL-10 and rotated MNIST), for each datum $\tilde{\bx}_i$ we sample an orbit size $\mu_i\sim\mathrm{Unif}\{6,\dots,10\}$ and generate variants $\bx_{ij}:=g_{ij}\cdot\tilde{\bx}_i$, with $g_{ij}\sim\bbQ$ and $\bbQ$ being the Haar measure on $G$. We attribute to $\bx_{ij}$ the same label as $\tilde{\bx}_i$. Here $G=C_4=\{0^\circ,90^\circ,180^\circ,270^\circ\}$ for CIFAR-10 and STL-10, and $G=G_{\text{MNIST}}=\{0^\circ,\pm10^\circ,\pm20^\circ,\pm30^\circ\}$ for MNIST. This yields orbits of variable size and a controlled level of redundancy. The values of $N_{\text{src}}$ for each dataset will be reported below. We always evaluate on the standard test split for each dataset.

\paragraph{\textit{Embedders and classifier}} 
We decompose the full classifier as $f_{\mathbf{w}} = c_\theta \circ q_\varphi$ with parameters $\mathbf{w}=(\varphi,\theta)$, where $c_\theta$ represents the classifier head and $q_\varphi$ denotes a feature extractor.
We use pretrained SimCLR backbones~\cite{chen2020simclr} (discarding the projection head) to construct both the function $h$ in GRINCO's quotient distance in \eqref{eq:d_G_psi} as well as a feature extractor $q_{\phi}$. More precisely, let $\phi_\varphi(\bx)\in\mathbb{R}^{C\times H\times W}$ be the final convolutional feature map of SimCLR and let $\nu$ denote global average pooling. The baseline embedding is $q_\varphi(\bx)=\nu(\phi_\varphi(\bx))$. For GRINCO, we additionally include an explicit orbit averaging to ensure $f_{\mathbf{w}}$ is exactly $G$-invariant:
\begin{equation}
  q_\varphi(\bx) = \frac{1}{|G|} \sum_{g\in G} \nu(\phi_\varphi(g\cdot\bx)).
  \label{eq:orbit-avg-practical}
\end{equation}
During AL, we keep $q_\varphi$ frozen for all methods, and construct a linear classifier head $c_\theta(\bz)$ on top of the frozen embedding $\bz=q_\varphi(\bx)$ by first performing PCA to improve efficiency, followed by a softmax layer that is trained on the labeled set at every AL iteration, yielding \mbox{$f_{\bw}(\bx)=\mathrm{softmax}(\bW \bz+\bb)$}.

\paragraph{\textit{Acquisition and coupling}} GRINCO computes $d_{G,h}$ from \eqref{eq:d_G_psi} with $h=q_\varphi$ from~\eqref{eq:orbit-avg-practical} and selects a batch by greedy $k$-center in this quotient space pseudo-distance, and model training is performed using the orbit-averaged objective in \eqref{eq:orbit-averaged-risk-coreset} with $\bbQ$ being the Haar measure on $G$. Since the number of samples closest to each orbit in the GRINCO coreset is not known, we use uniform weights $w_i=1/K$ for the labeled orbits. This is a practical approximation of the coreset weights used in the theoretical analysis. 
We compare GRINCO against uniform \emph{random} sampling, \emph{entropy} and \emph{margin}~\cite{settles2009} (which are uncertainty-based), \emph{Euclidean} $k$-centers~\cite{sener2018coresets}, and BADGE~\cite{ash2020badge}. We also include a simple variant of the proposed approach, which we call GRINCO-b, by replacing $k$-centers with the $k$-means objective in the quotient metric (i.e., the coreset is constructed from the cluster centers). The $k$-means criterion encourages a more balanced coverage compared to $k$-centers, and also serves to illustrate how our orbit-level coreset formulation can be coupled with different coverage criteria.
Each run starts from an initial labeled set of size $F_0$, acquires $b$ new labels per round for $T$ rounds, reaching a total budget of $B=F_0+Tb$. We report the dataset-specific $F_0$, $b$, and $T$ below. 

\paragraph{\textit{Metrics, runs, and dataset versions}} We report overall test classification accuracy, label-efficiency curves (accuracy versus amount of queried labels), and orbit efficiency $\eta$ (proportion of unique orbits in the coreset).
Each simulation is repeated for $R=5$ Monte Carlo runs. For the \emph{rotated} datasets, which include additional rotated samples, we additionally evaluate the methods on five different random realizations of the dataset (indexed by $v\in\{0,1,2,3,4\}$), yielding a total of $25$ experiments per configuration. For standard datasets, we report mean $\pm$ std over the $R$ runs. For rotated variants, we first average over the $R$ runs for each dataset version, and then report mean $\pm$ std across the $v$ dataset versions to capture the variability across different random realizations of the pool.

\subsubsection{\textbf{Parameter settings}}\label{sec:exp:dataset-specific}

The parameters for each dataset and budget settings $(F_0,b,T,B)$ are as follows.

\paragraph{\textit{CIFAR-10}} 50k train and 10k test images, of size $32\times32\times3$. Pools with explicit redundancy use $N_{\text{src}}=2000$ class-balanced initial images $\tilde \bx_i$ with $C_4$ rotations, from which the different samples are generated as described previously. Backbone: SimCLR-pretrained ResNet-18 for CIFAR-10 ($3\times3$ conv1, no max pooling), final feature map $512\times4\times4$, embedding dimension $512$, PCA dimension $64$. AL setting: $(100,100,49,5000)$.

\paragraph{\textit{STL-10}} 5k labeled train, 8k test, and 100k unlabeled images, all $96\times96\times3$. Pools with explicit redundancy use $N_{\text{src}}=1000$ class-balanced initial images $\tilde \bx_i$ with $C_4$ rotations, from which the different samples are generated as described previously. Backbone: SimCLR-pretrained ResNet-18 for $96\times96$, final feature map $512\times3\times3$, embedding dimension $512$, PCA dimension $64$. AL setting: $(50,50,29,1500)$.

\paragraph{\textit{MNIST}} 60k train and 10k test images, $28\times28\times1$. To preserve digit semantics, we use $G_{\text{MNIST}}=\{0^\circ,\pm10^\circ,\pm20^\circ,\pm30^\circ\}$ instead of $C_4$. These rotations use bilinear interpolation with zero-padding and center-cropping. Pools for the rotated MNIST variations use $N_{\text{src}}=2000$ class-balanced initial images $\tilde \bx_i$ with $G_{\text{MNIST}}$ rotations, from which the different samples are generated as described previously. Backbone: SimCLR-pretrained ConvNet with two $5\times5$ convolutional blocks plus max pooling, final feature map $64\times4\times4$, embedding dimension $128$, and PCA dimension $8$. AL setting: $(10,10,49,500)$.

\subsubsection{\textbf{Results: Labeling efficiency under rotation-induced symmetry}}\label{sec:exp:random-query}

Before presenting the main AL results, we first show a simpler simulation that shows the impact of orbit-level coreset selection in classification performance, without the effect of the iterative AL pipeline. To this end, we consider a sample of the rotated CIFAR-10 dataset variant ($v=0$) described earlier (with $N_{\rm src}=2000$ and $N=18012$ total samples), and compare classification accuracy between GRINCO and uniform random sampling, for coreset sizes of $500$ and $2000$ labels. Note that $N_{\text{src}}$ is the maximum number of samples from unique orbits in the dataset.

Results are reported in Table~\ref{tab:simquery-cifar10}. It can be seen that GRINCO achieves nearly 100\% orbit efficiency, meaning that almost all selected labels correspond to unique orbits, while random sampling achieves 88.7\% orbit efficiency at 500 labels and about 65.1\% at 2000 labels. This translates into a higher overall classification accuracy (OA) for GRINCO, with 0.56\% improvement at 500 labels and 0.99\% improvement at 2000 labels, compared to random sampling. This shows how the proposed method can take into account symmetries in the dataset to improve the selection of samples and the use of labeling budget.
Figure~\ref{fig:aug-sampling-efficiency} shows OA and orbit efficiency $\eta$ as a function of the number of labels for the same simulation setting, from which it can be seen that the proposed strategy yields consistently higher OA and orbit efficiency across the range of labels. Moreover, the orbit efficiency of GRINCO only starts to decrease after 2000 labels, when the number of unique orbits in the dataset is exhausted.

\begin{table}[t]
    \centering
    \setlength{\tabcolsep}{4pt}
    \caption{Labeling-efficiency results on rotated CIFAR-10: orbit coverage and overall accuracy (OA) over $R=5$ runs.}
    \label{tab:simquery-cifar10}
    \begin{tabular}{@{}c l r r r@{}}
        \toprule
        \multicolumn{1}{c}{Labels} & \multicolumn{1}{c}{Strategy} & \multicolumn{1}{c}{\# unique orbits} & \multicolumn{1}{c}{$\eta$ (\%)} & \multicolumn{1}{c}{Test OA (\%)} \\
        \midrule
        \multirow{2}{*}{500} & GRINCO & $\mathbf{497\pm2}$ & $\mathbf{99.4\pm0.3}$ & $\mathbf{75.57\pm0.27}$ \\
        & Random & $443\pm3$ & $88.7\pm0.6$ & $75.01\pm0.29$ \\
        \midrule
        \multirow{2}{*}{2000} & GRINCO & $\mathbf{1997\pm2}$ & $\mathbf{99.9\pm0.1}$ & $\mathbf{78.43\pm0.21}$ \\
        & Random & $1302\pm13$ & $65.1\pm0.6$ & $77.44\pm0.16$ \\
        \bottomrule
    \end{tabular}
\end{table}

\begin{figure}[htbp]
    \centering
    \subfloat[]{%
      \includegraphics[width=0.40\textwidth]{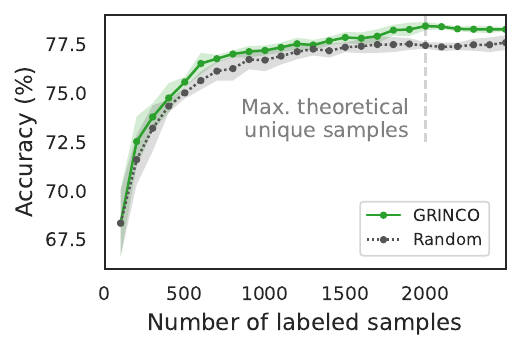}%
      \label{fig:aug-sampling-efficiency-a}%
    }
    
    \subfloat[]{%
      \includegraphics[width=0.40\textwidth]{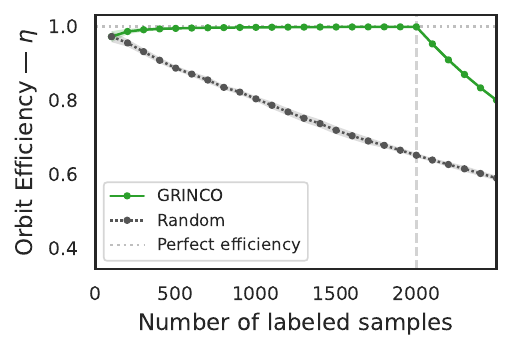}%
      \label{fig:aug-sampling-efficiency-b}%
    } \\
    \caption{Labeling-efficiency results on rotated CIFAR-10 ($v=0$) for GRINCO and random sampling. (a) Test overall accuracy (OA) as the labeling budget increases. (b) Orbit efficiency $\eta$, showing the fraction of queried labels that correspond to unique orbits.}
    \label{fig:aug-sampling-efficiency}
\end{figure}

\subsubsection{\textbf{Results: Full Active Learning pipeline}}\label{sec:exp:al-results}

This section reports the results of the experiments with the full AL pipeline described above on the CIFAR-10, STL-10, and MNIST datasets, as well as their rotated variants. We compare GRINCO and GRINCO-b to the baselines in terms of overall accuracy and orbit efficiency at a given budget, as well as their trajectories across AL iterations. 
Table~\ref{tab:al-summary} reports results for all datasets and methods for budgets of $B=2000$, $1500$, and $500$ labels for CIFAR-10, STL-10, and MNIST, respectively, while Figure~\ref{fig:aug-acc-orbits} shows the accuracy and orbit efficiency trajectories for the rotated CIFAR-10 dataset (STL-10 follows similar trends, and MNIST is discussed below).

\begin{table*}[t]
    \centering
    \setlength{\tabcolsep}{2pt}
    \caption{Active Learning results for all datasets and methods. We report mean $\pm$ std test OA (\%) and, for rotated variants, orbit efficiency $\eta$ (\%). The numbers of labeled samples are $2000$ for CIFAR-10, $1000$ for STL-10, and $500$ for MNIST. Best results are marked in bold.}
    \label{tab:al-summary}
    \begin{tabular}{l c cc c cc c cc}
        \toprule
        \multirow{2}{*}{Method} & \multicolumn{1}{c}{CIFAR-10} & \multicolumn{2}{c}{Rotated CIFAR-10} & \multicolumn{1}{c}{STL-10} & \multicolumn{2}{c}{Rotated STL-10} & \multicolumn{1}{c}{MNIST} & \multicolumn{2}{c}{Rotated MNIST} \\
        \cmidrule(lr){2-2} \cmidrule(lr){3-4} \cmidrule(lr){5-5} \cmidrule(lr){6-7} \cmidrule(lr){8-8} \cmidrule(lr){9-10}
        & OA (\%) & OA (\%) & $\eta$ (\%) & OA (\%) & OA (\%) & $\eta$ (\%) & OA (\%) & OA (\%) & $\eta$ (\%) \\
        \midrule
        Random & $78.4\pm0.4$ & $76.8\pm0.1$ & $65.0\pm0.1$ & $61.8\pm0.6$ & $57.8\pm0.8$ & $65.3\pm0.2$ & $95.9\pm0.2$ & $95.5\pm0.1$ & $89.5\pm0.4$ \\
        Entropy & $79.4\pm0.4$ & $74.4\pm0.4$ & $23.9\pm0.2$ & $62.9\pm0.4$ & $53.5\pm0.8$ & $31.4\pm0.7$ & $95.1\pm0.5$ & $94.4\pm1.0$ & $32.8\pm0.7$ \\
        Margin & $\boldsymbol{80.5\pm0.1}$ & $76.1\pm0.2$ & $31.6\pm0.2$ & $\boldsymbol{63.3\pm0.5}$ & $54.8\pm1.2$ & $37.8\pm0.6$ & $\boldsymbol{96.9\pm0.1}$ & $96.3\pm0.2$ & $44.3\pm0.6$ \\
        BADGE & $80.2\pm0.2$ & $77.7\pm0.2$ & $48.6\pm0.6$ & $63.1\pm0.3$ & $58.3\pm1.1$ & $60.5\pm0.8$ & $96.6\pm0.1$ & $96.4\pm0.2$ & $75.1\pm0.6$ \\
        Euclidean & $77.0\pm0.5$ & $78.7\pm0.2$ & $97.1\pm0.3$ & $61.3\pm0.3$ & $61.2\pm1.0$ & $99.4\pm0.2$ & $93.0\pm2.0$ & $94.4\pm1.7$ & $95.4\pm0.5$ \\
        GRINCO & $78.2\pm0.7$ & $\boldsymbol{79.7\pm0.3}$ & $\boldsymbol{99.9\pm0.0}$ & $62.4\pm0.5$ & $\boldsymbol{62.6\pm1.0}$ & $\boldsymbol{99.9\pm0.0}$ & $96.2\pm0.5$ & $95.8\pm1.0$ & $\boldsymbol{97.1\pm0.4}$ \\
        GRINCO-b & $79.2\pm0.4$ & $79.6\pm0.3$ & $99.1\pm0.0$ & $63.0\pm0.6$ & $62.4\pm0.9$ & $99.1\pm0.0$ & $96.8\pm0.1$ & $\boldsymbol{96.5\pm0.2}$ & $96.5\pm0.3$ \\
        \bottomrule
    \end{tabular}
\end{table*}

On the standard CIFAR-10 and STL-10 datasets, the uncertainty baselines achieve slightly stronger performance compared to the coreset approaches at the budgets reported in Table~\ref{tab:al-summary}. However, GRINCO and GRINCO-b are still competitive and consistently improve over the Euclidean $k$-centers coreset approach. On CIFAR-10, GRINCO-b improves OA by about $2.9\%$ relative to \emph{Euclidean}, and it stays within about $1.6\%$ relative to \emph{margin}, while on STL-10 it improves OA by about $2.8\%$ relative to \emph{Euclidean} and achieves performance that is within one standard deviation of \emph{margin}'s results.
For the rotated variants of the datasets, the performance of the proposed GRINCO and GRINCO-b methods is significantly stronger than the baselines. On rotated CIFAR-10, GRINCO improves OA by about $1.3\%$ relative to \emph{Euclidean} and by about $2.6\%$ relative to BADGE, and GRINCO-b performs similarly. On rotated STL-10, GRINCO improves OA by about $2.3\%$ relative to \emph{Euclidean} and by about $7.4\%$ relative to BADGE. Orbit efficiency follows the same trend, with GRINCO reaching nearly $100\%$ samples coming from unique orbits, which is higher than \emph{Euclidean}, and considerably higher than BADGE, \emph{margin}, \emph{entropy} and \emph{random}, all of which have orbit efficiency below $66\%$.
This illustrates the benefits of the proposed orbit-based coreset selection on datasets with explicit rotation-induced symmetry.

The results on MNIST presented in Table~\ref{tab:al-summary} illustrate the behavior of GRINCO on image data with the smaller-angle rotation group $G_{\text{MNIST}}$, instead of $C_4$. 
On standard MNIST, all methods achieve similar OA, with \emph{margin} and GRINCO-b (the top performing methods) within one standard deviation of each other.
On the rotated MNIST variant, OA is again tightly clustered, with GRINCO-b at $96.5$\%, BADGE at $96.4$\%, \emph{margin} at $96.3$\%, and GRINCO at $95.8$\%, which perform significantly better than the Euclidean $k$-centers. Orbit efficiency ($\eta$) is very high for the coreset approaches, as GRINCO reaches $97.1$\%, GRINCO-b $96.5$\%, and \emph{Euclidean} $95.4$\%, while BADGE, \emph{margin}, and \emph{entropy} fall to $75.1$\%, $44.3$\%, and $32.8$\%, respectively. These results show that quotient space coreset selection is most useful when redundancy in the dataset is high, where it achieves very high orbit coverage. 

The overall results indicate that uncertainty-based acquisition is very strong in low-redundancy regimes, however, orbit-aware methods remain very competitive despite being based on a different, purely geometric principle that does not take into account any information about the data labels, the classifier output, or its uncertainty. This emphasizes the generality of the coreset approach which, as a data summarization technique, is not restricted to a classification task.
This behavior is consistent with the quotient-space perspective developed in the theoretical formulation. When the pool contains many samples within the same orbit, methods that do not operate on the quotient space can repeatedly query points with the same semantic content, wasting labeling budget and reducing coverage. Instead, GRINCO and GRINCO-b use a quotient metric that collapses each orbit to a single point. Thus, their most important improvements appear in datasets with rotation-induced symmetry from a known group $G$, thereby improving orbit efficiency and yielding strong OA.

The baselines show complementary behavior across regimes. Random sampling is competitive only when the dataset does not have redundancy, while \emph{entropy} and \emph{margin} show strong performance on standard datasets but do not explicitly control orbit coverage, as they are based on different operating principles that target AL classification. BADGE and Euclidean $k$-centers have better orbit coverage, yet both still allocate labels to redundant samples. 
Both GRINCO and GRINCO-b achieve high OA and very high orbit efficiency. However, GRINCO-b's use of the $k$-means criterion promotes balanced coverage across orbits, which was shown to be beneficial in the CIFAR-10 and MNIST datasets. Further refining the quotient space selection criterion to take into account uncertainty or other task-specific information is a promising direction for future work, as it could leverage the complementary advantages of both approaches.

\begin{figure}[htbp]
    \centering
    \subfloat[]{%
        \includegraphics[width=0.48\textwidth]{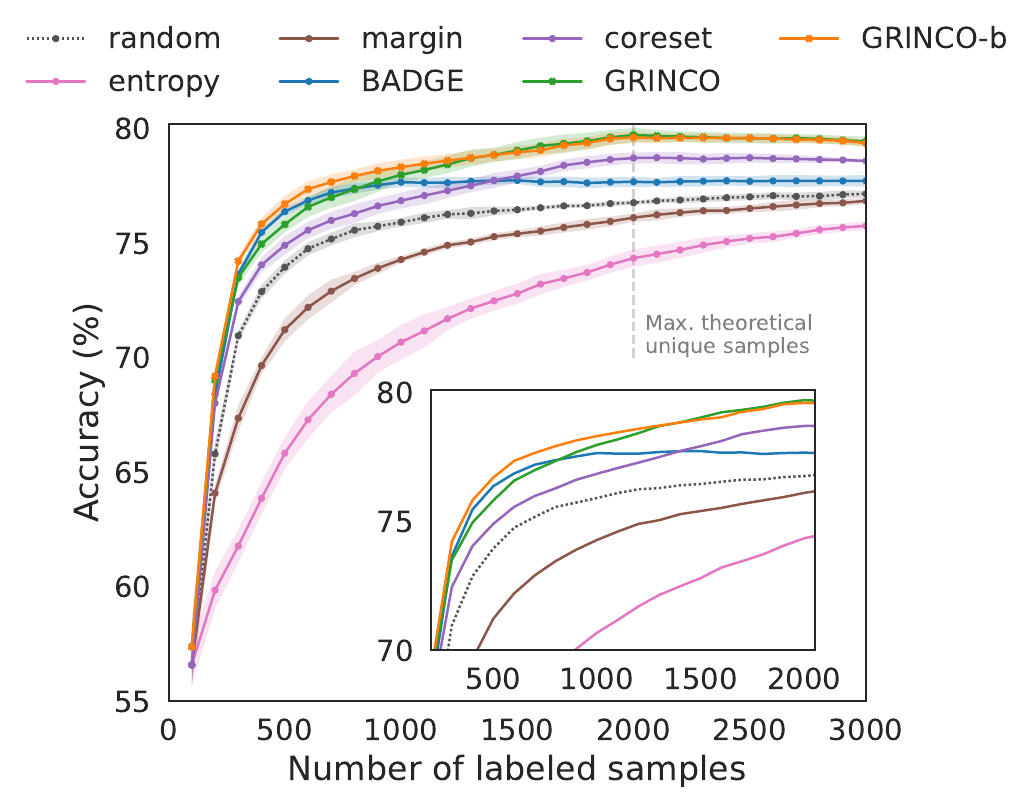}%
        \label{fig:aug-acc-orbits-acc}%
    }

    \subfloat[]{%
        \includegraphics[width=0.48\textwidth]{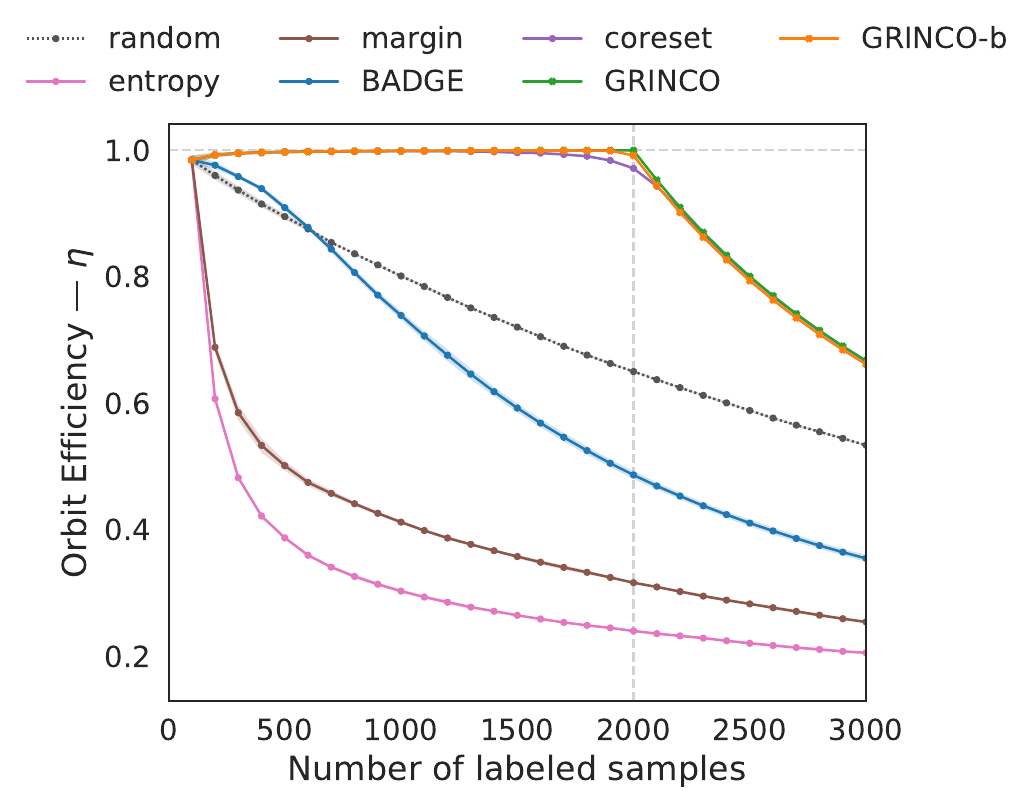}%
        \label{fig:aug-acc-orbits-orbits}%
    } \\
    \caption{Full active-learning trajectories on rotated CIFAR-10. (a) Test overall accuracy (OA) for the acquisition strategies as the number of labeled samples increases. (b) Orbit efficiency $\eta$ for the same strategies, showing how much of the labeling budget is allocated to unique orbits.}
    \label{fig:aug-acc-orbits}
\end{figure}

\section{Conclusion}\label{sec:conclusion}

We introduced a group-invariant coreset framework for AL that formulates a coverage criterion on the quotient space induced by the action of a transformation group. The method performs sample selection at the level of orbits rather than raw samples, exploiting appropriate metrics constructed from orbit-separating invariants. Model training in AL is performed using an orbit-averaged classification objective. This yields a representative coreset that avoids selecting and labeling samples which are redundant under the group symmetry. We provided theoretical analysis linking coreset coverage in the quotient space to excess risk in AL, and experimentally evaluated the approach on a synthetic \emph{rays} dataset as well as in real image benchmarks with rotation-induced symmetry. The results show that the proposed GRINCO reduces redundancy in the selected coreset and can significantly improve labeling efficiency when invariances are present in the data.
This work assumes knowledge of a group action and access to a quotient space distance or orbit-separating map. Extending the approach to settings with unknown or approximate data symmetries is an important direction for future work.

\section*{Acknowledgment}
{\footnotesize This work was supported in part by Coordena\c{c}\~{a}o de Aperfei\c{c}oamento de Pessoal de N\'{i}vel Superior (CAPES), Brazil, Finance Code 001; the French National Research Agency (ANR) under Grants ANR-23-CE23-0024, ANR-23-CE94-0001, and ANR-25-CE23-0949; and the Brazilian Research Council under Grant 304597/2023-6.\par}

\begingroup
\renewcommand{\footnotesize}{\fontsize{7.6pt}{8.6pt}\selectfont}
\bibliographystyle{IEEEtran}
\bibliography{refs}
\endgroup

\clearpage
\onecolumn
\appendix[Proof of the Generalization Theorem]\label{app:generalization-proof}
\setcounter{equation}{0}
\renewcommand{\theequation}{A.\arabic{equation}}
\renewcommand{\theHequation}{A.\arabic{equation}}

This proof of Theorem~\ref{thm:generalization_gcoresets} complements Subsection~\ref{subsec:generalization} by bounding the generalization gap via a decomposition into a full-dataset term and an invariant-coreset term.

\begin{proof}
We can decompose the generalization error as
\begin{align}
    & \underbrace{\big|\hat{\mathscr{R}}_{\mathcal{C}}^{\bbQ}(f_{\bw}) - \mathscr{R}(f_{\bw})\big|}_{\text{coreset generalization error}}
    \leq \underbrace{\big|\hat{\mathscr{R}}_{\mathrm{full}}^{\bbQ}(f_{\bw}) - \mathscr{R}(f_{\bw})\big|}_{\text{full dataset generalization term}} + \underbrace{\big|\hat{\mathscr{R}}_{\mathcal{C}}^{\bbQ}(f_{\bw}) - \hat{\mathscr{R}}_{\mathrm{full}}^{\bbQ}(f_{\bw})\big|}_{\text{invariant coreset excess risk}} \,.
\end{align}
Thus, we see that the generalization error includes an excess risk due to the use of the representative coreset $\mathcal{C}$ compared to the full dataset, given by $|\hat{\mathscr{R}}_{\mathcal{C}}^{\bbQ}(f_{\bw}) - \hat{\mathscr{R}}_{\mathrm{full}}^{\bbQ}(f_{\bw})|$.
Expanding the last term, without assuming a constant label per orbit:
\begin{align}
    & |\hat{\mathscr{R}}_{\mathcal{C}}^{\bbQ}(f_{\bw}) - \hat{\mathscr{R}}_{\mathrm{full}}^{\bbQ}(f_{\bw})| \nonumber \\
    & = \Big|\frac{1}{N} \sum_{n=1}^N \bbE_{\bbQ(g)}\{\mathcal{L}(f_{\bw}(g\cdot\bx_n),y_n)\} - \sum_{i=1}^{K} w_i \bbE_{\bbQ(g)}\{\mathcal{L}(f_{\bw}(g\cdot\bar{\bx}_i),y_i)\} \Big| \nonumber \\
    & \leq \Big|\frac{1}{N} \sum_{n=1}^N \bbE_{\bbQ(g)}\{\mathcal{L}(f_{\bw}(g\cdot\bx_n),y_n)\} - \frac{1}{N} \sum_{n=1}^N \bbE_{p(Y|g\cdot \bx_n) \bbQ(g) }\{\mathcal{L}(f_{\bw}(g\cdot \bx_n),Y)\} \Big| \nonumber \\
    & \quad + \Big|\sum_{i=1}^{K} w_i \bbE_{\bbQ(g)}\{\mathcal{L}(f_{\bw}(g\cdot\bar{\bx}_i),y_i)\} - \frac{1}{N} \sum_{n=1}^N \bbE_{p(Y|g\cdot \bx_n) \bbQ(g) }\{\mathcal{L}(f_{\bw}(g\cdot \bx_n),Y)\}\Big|
    \label{eq:thm_eq_part1_twoterms}
\end{align}
We now have to bound each of the terms in the above equation.

\textbf{First term:}
To upper bound the first term in \eqref{eq:thm_eq_part1_twoterms}, we make use of Bernstein's inequality and the boundedness of the loss, i.e., the fact that
\begin{align}
    & \hspace{-10ex} |\bbE_{\bbQ(g)}\{\mathcal{L}(f_{\bw}(g\cdot\bx_n),y_n)\}-\bbE_{p(Y|g\cdot \bx_n) \bbQ(g) }\{\mathcal{L}(f_{\bw}(g\cdot \bx_n),Y)\}| \nonumber 
    \\
    & \leq \max_{y\in\mathcal{Y}, \, n\in\{1,\ldots,N\}} \bbE_{\bbQ(g)}\{\mathcal{L}(f_{\bw}(g\cdot\bx_n),y)\} \nonumber \\
    & \leq L_{\max} \,. \nonumber 
\end{align}

Since $(\bx_n,y_n)$ are i.i.d., we can define the random variable $Z$ and its i.i.d. realizations $Z_n$ as 
\begin{align}
    Z & = \bbE_{\bbQ(g)}\{\mathcal{L}(f_{\bw}(g\cdot \bX),Y)\} - \bbE_{p(Y'|g\cdot \bX) \bbQ(g) }\{\mathcal{L}(f_{\bw}(g\cdot \bX),Y')\}
    \\
    Z_n & =\bbE_{\bbQ(g)}\{\mathcal{L}(f_{\bw}(g\cdot\bx_n),y_n)\} - \bbE_{p(Y|g\cdot \bx_n) \bbQ(g) }\{\mathcal{L}(f_{\bw}(g\cdot \bx_n),Y)\}
\end{align}
Note that $Z\equiv Z(\bX,Y)$ is a function of $\bX$ and $Y$, and can also be written as the orbit-average of a random variable $\widetilde{Z}(\bX,Y)$ as
\[
Z = \bbE_{\bbQ(g)}\{\widetilde{Z}(g\cdot \bX, Y)\}, \quad \text{ where } \quad  \widetilde{Z}(\bX,Y)=\mathcal{L}(f_{\bw}(\bX),Y) - \bbE_{p(Y'|\bX)}\{\mathcal{L}(f_{\bw}(\bX),Y')\} \,.
\]

One can verify that $Z$ is zero mean due to the group invariance property of $p(Y|\bX)$, that is,
\begin{align}
    & \bbE_{p(Y|\bX)p(\bX)}\{Z\} \nonumber \\
    & = \bbE_{p(Y|\bX)p(\bX)}\Big\{\bbE_{\bbQ(g)}\{\mathcal{L}(f_{\bw}(g\cdot \bX),Y)\} - \bbE_{p(Y'|g\cdot \bX) \bbQ(g) }\{\mathcal{L}(f_{\bw}(g\cdot \bX),Y')\}\Big\} \nonumber \\
    & = \bbE_{p(\bX)}\Big\{\bbE_{p(Y|\bX)}\Big\{\bbE_{\bbQ(g)}\{\mathcal{L}(f_{\bw}(g\cdot \bX),Y)\}\Big\} - \bbE_{p(Y'|\bX) \bbQ(g) }\{\mathcal{L}(f_{\bw}(g\cdot \bX),Y')\}\Big\} \nonumber \\ 
    & = 0 \,. \nonumber
\end{align}
Applying Bernstein's inequality~\cite{maurer2019bernsteinBoundedInteraction,sridharan2002gentleConcentrationIneq} to bound $\frac{1}{N}\sum_{n=1}^N Z_n$ gives, with probability at least $1-\gamma$,
\begin{align}
    & \Big|\frac{1}{N} \sum_{n=1}^N \bbE_{\bbQ(g)}\{\mathcal{L}(f_{\bw}(g\cdot\bx_n),y_n)\} - \frac{1}{N} \sum_{n=1}^N \bbE_{p(Y|g\cdot \bx_n) \bbQ(g) }\{\mathcal{L}(f_{\bw}(g\cdot \bx_n),Y)\} \Big| \nonumber \\
    & \le
    \sqrt{ \frac{2\sigma_Z^2 \ln(2/\gamma)}{N} }
    + \frac{2\,L_{\max}\,\ln(2/\gamma)}{3N}.
    \label{eq:thm_bernstein_1}
\end{align}
where $\sigma_Z^2:=\mathrm{Var}(Z)$ is the variance of $Z$. Since $Z\equiv Z(\bX,Y)$ is a function of $\bX$ and $Y$, it can be expressed using the law of total variance as
\begin{align}
    \mathrm{Var}_{p(Y,\bX)}(Z(\bX,Y)) & = \mathrm{Var}_{p(Y,\bX)}(\bbE_{\bbQ(g)}\{\widetilde{Z}(g\cdot \bX,Y)\})
    \nonumber\\
    & = \bbE_{p(\bX)}\{\mathrm{Var}_{p(Y|\bX)}(Z(\bX,Y)|\bX)\} + \mathrm{Var}_{p(\bX)}(\bbE_{p(Y|\bX)}\{Z(\bX,Y)\}) \,.
    \label{eq:thm_total_var_1}
\end{align}
Denote the following functions to lighten the notation:
\begin{align}
    \zeta(\bX) & = \bbE_{\bbQ(g)}\{\bbE_{p(Y|\bX)}\{\widetilde{Z}(g\cdot \bX,Y)\}\}
    \\
    \widetilde{\zeta}(\bX) & = \bbE_{p(Y|\bX)}\{\widetilde{Z}(\bX,Y)\}
    \\
    \xi & = \bbE_{p(\bX)}\{\mathrm{Var}_{p(Y|\bX)}(Z(\bX,Y)|\bX)\}
\end{align}
First, for the term $\xi$ (which is the first term in \eqref{eq:thm_total_var_1}), using Jensen's inequality in measure-theoretic form, and the fact that $Z$ is (conditionally) zero mean (i.e., $\bbE_{p(Y'|g\cdot \bX)}\{Z\}=0$ for any $g$) which makes the map $Z\mapsto \mathrm{Var}_{p(Y|\bX)}(Z|\bX)$ convex, we can write
\begin{align}
    0 \le \xi & = \bbE_{p(\bX)}\{\mathrm{Var}_{p(Y|\bX)}(Z(\bX,Y)|\bX)\} \nonumber \\
    & = \bbE_{p(\bX)}\{\mathrm{Var}_{p(Y|\bX)}\big(\bbE_{\bbQ(g)}\{\widetilde{Z}(g\cdot \bX, Y)\}|\bX\big)\} \nonumber \\
    & \le \bbE_{p(\bX)}\{\bbE_{\bbQ(g)}\{\mathrm{Var}_{p(Y|\bX)}\big(\widetilde{Z}(g\cdot \bX, Y)|\bX\big)\}\} \nonumber \\
    & := \bbE_{p(\bX)}\{\bbE_{\bbQ(g)}\{ \widetilde{\xi}(g\cdot \bX) \}\}
    \label{eq:thm_group_inv_1}
\end{align}
where we used the group invariance property $p(Y|\bX)=p(Y|g\cdot \bX)$, and $\widetilde{\xi}$ is given by
\[\widetilde{\xi}(\bX) = \mathrm{Var}_{p(Y|\bX)}\big(\widetilde{Z}(\bX, Y)|\bX\big) \,.\]

For the second term in \eqref{eq:thm_total_var_1}, using Lemma 1 in \cite{chen2020group}, we also have
\begin{align}
    \mathrm{Var}_{p(\bX)}(\bbE_{p(Y|\bX)}\{Z(\bX,Y)\}) & = \mathrm{Var}_{p(\bX)}(\bbE_{\bbQ(g)}\{\bbE_{p(Y|\bX)}\{\widetilde{Z}(g\cdot \bX,Y)\}\}) \nonumber \\
    & = \mathrm{Var}_{p(\bX)}(\zeta(\bX)) \nonumber \\
    & = \mathrm{Var}_{p(\bX)}(\widetilde{\zeta}(\bX)) - \bbE_{p(\bX)} \{\mathrm{Var}_{\bbQ(g)}(\widetilde{\zeta}(g\cdot \bX))\} \,,
    \label{eq:thm_group_inv_2}
\end{align}
plugging \eqref{eq:thm_group_inv_1} and \eqref{eq:thm_group_inv_2} back into \eqref{eq:thm_bernstein_1}, 
\begin{align}
    & \Big|\frac{1}{N} \sum_{n=1}^N \bbE_{\bbQ(g)}\{\mathcal{L}(f_{\bw}(g\cdot\bx_n),y_n)\} - \frac{1}{N} \sum_{n=1}^N \bbE_{p(Y|g\cdot \bx_n) \bbQ(g) }\{\mathcal{L}(f_{\bw}(g\cdot \bx_n),Y)\} \Big| \nonumber \\
    & \le \sqrt{ \frac{2 \ln(2/\gamma)}{N} V(\bbQ) }
    + \frac{2\,L_{\max}\,\ln(2/\gamma)}{3N}.
    \label{eq:thm_bernstein_2}
\end{align}
where
\begin{align}
    V(\bbQ) = \bbE_{p(\bX)}\{\bbE_{\bbQ(g)}\{ \widetilde{\xi}(g\cdot \bX) \}\} + \mathrm{Var}_{p(\bX)}(\widetilde{\zeta}(\bX)) - \bbE_{p(\bX)} \{\mathrm{Var}_{\bbQ(g)}(\widetilde{\zeta}(g\cdot \bX))\}
    \nonumber
\end{align}
is the variance term as defined in \eqref{eq:thm_variance_term}.
Note that for large $N$ the first term dominates the convergence speed.

\textbf{Second term:}
For the second term, we note that for every point $\bx_n$ there is a representative in the coreset that is $\varepsilon$-close. Thus, assuming without loss of generality that each $\bx_n$ has a unique nearest neighbor in the coreset (consistent with the construction of the representative coreset, where the selected samples belong to disjoint orbits), we can write:
\begin{align}
    & \Big|\sum_{i=1}^{K} w_i \bbE_{\bbQ(g)}\{\mathcal{L}(f_{\bw}(g\cdot\bar{\bx}_i),y_i)\} \,\,-\,\, \frac{1}{N} \sum_{n=1}^N \bbE_{p(Y|g\cdot \bx_n) \bbQ(g) }\{\mathcal{L}(f_{\bw}(g\cdot \bx_n),Y)\}\Big| \nonumber \\
    & = \Big|\sum_{i=1}^{K} \Big(w_i\bbE_{\bbQ(g)}\{\mathcal{L}(f_{\bw}(g\cdot\bar{\bx}_i),y_i)\}  \,\,-\,\, \frac{1}{N} \sum_{j : [\bar{\bx}_i]_G \text{ is closest to } [\bx_j]_G} \bbE_{p(Y|g\cdot \bx_j) \bbQ(g) }\{\mathcal{L}(f_{\bw}(g\cdot \bx_j),Y)\}\Big)\Big|  \,.
    \label{eq:thm_eq_2a} 
\end{align}
Using an approach similar to that in \cite{sener2018coresets}, we can write,
\begin{align}
    & \bbE_{p(Y|g\cdot \bx_j) \bbQ(g) }\{\mathcal{L}(f_{\bw}(g\cdot \bx_j),Y)\} \nonumber \\
    & = \bbE_{\bbQ(g) }\Big\{\sum_{c\in\mathcal{Y}} p(Y=c|g\cdot \bar{\bx}_i) \mathcal{L}(f_{\bw}(g\cdot \bar{\bx}_i),c) \nonumber \\
    & \qquad + p(Y=c|g\cdot \bar{\bx}_i) \big(\mathcal{L}(f_{\bw}(g\cdot \bx_j),c) - \mathcal{L}(f_{\bw}(g\cdot \bar{\bx}_i),c) \big)
    \nonumber \\
    & \qquad + \big(p(Y=c|g\cdot \bx_j) - p(Y=c|g\cdot \bar{\bx}_i)\big) \mathcal{L}(f_{\bw}(g\cdot \bx_j),c)
    \Big\} 
\end{align}
We now use the fact that the labeling function and loss function are Lipschitz continuous with respect to the quotient pseudo-distance in \eqref{eq:thm_a_lip_p} and~\eqref{eq:thm_a_lip_L}, that is,
\begin{align}
    |p(Y=c|\bx_j) - p(Y=c|\bar{\bx}_i)| & \leq L_{p} d_{G,h}(\bx_j,\bar{\bx}_i)
    \\
    |\bbE_{\bbQ(g) }\big\{\mathcal{L}(f_{\bw}(g\cdot \bx_j),Y)\big\} - \bbE_{\bbQ(g) }\big\{\mathcal{L}(f_{\bw}(g\cdot \bar{\bx}_i),Y)\big\} |  & \leq L_{\mathcal{L}} d_{G,h}(\bx_j,\bar{\bx}_i)
\end{align}
and that the loss is bounded as $\mathcal{L}(y,\hat{y})\le L_{\max}$ over the domain. Using the per-class bound and summing over $c$ gives $\sum_{c\in\mathcal{Y}} |p(Y=c|\bx_j)-p(Y=c|\bar{\bx}_i)| \le C L_p d_{G,h}(\bx_j,\bar{\bx}_i)$. Thus, the triangle inequality gives us
\begin{align}
    & \Big|\sum_{i=1}^{K} \Big(w_i\bbE_{\bbQ(g)}\{\mathcal{L}(f_{\bw}(g\cdot\bar{\bx}_i),y_i)\} - \frac{1}{N} \sum_{j : [\bar{\bx}_i]_G \text{ is closest to } [\bx_j]_G} \bbE_{p(Y|g\cdot \bx_j) \bbQ(g) }\{\mathcal{L}(f_{\bw}(g\cdot \bx_j),Y)\}\Big)\Big| \nonumber \\
    & \leq \Big|\sum_{i=1}^{K} \Big(w_i\bbE_{\bbQ(g)}\{\mathcal{L}(f_{\bw}(g\cdot\bar{\bx}_i),y_i)\} - \frac{N_i}{N} \bbE_{p(Y|g\cdot \bar{\bx}_i)\bbQ(g)}\{\mathcal{L}(f_{\bw}(g\cdot\bar{\bx}_i),Y)\} \Big)\Big| \nonumber \\
    & \qquad + \Big|\frac{1}{N} \sum_{i=1}^{K} \sum_{j : [\bar{\bx}_i]_G \text{ is closest to } [\bx_j]_G} \bbE_{\bbQ(g)}\Big\{ 
    \sum_{c\in\mathcal{Y}} p(Y=c|g\cdot \bar{\bx}_i) \big(\mathcal{L}(f_{\bw}(g\cdot \bx_j),c) - \mathcal{L}(f_{\bw}(g\cdot \bar{\bx}_i),c) \big)
    \nonumber \\
    &\qquad + \big(p(Y=c|g\cdot \bx_j) - p(Y=c|g\cdot \bar{\bx}_i)\big) \mathcal{L}(f_{\bw}(g\cdot \bx_j),c)
    \Big\} \Big| \nonumber \\
    & \leq \Big|\sum_{i=1}^{K} \Big(w_i\bbE_{\bbQ(g)}\{\mathcal{L}(f_{\bw}(g\cdot\bar{\bx}_i),y_i)\} - \frac{N_i}{N} \bbE_{p(Y|g\cdot \bar{\bx}_i)\bbQ(g)}\{\mathcal{L}(f_{\bw}(g\cdot\bar{\bx}_i),Y)\}  \Big)\Big| \nonumber \\
    & \qquad + \frac{1}{N} \sum_{i=1}^{K} \sum_{j : [\bar{\bx}_i]_G \text{ is closest to } [\bx_j]_G} \bbE_{\bbQ(g)}\Big\{ 
    L_{\mathcal{L}} d_{G,h}(\bx_j,\bar{\bx}_i)  + C L_{p} L_{\max} d_{G,h}(\bx_j,\bar{\bx}_i) \Big\} \nonumber \\
    & \leq \Big|\sum_{i=1}^{K} \Big(w_i\bbE_{\bbQ(g)}\{\mathcal{L}(f_{\bw}(g\cdot\bar{\bx}_i),y_i)\} - \frac{N_i}{N} \bbE_{p(Y|g\cdot \bar{\bx}_i)\bbQ(g)}\{\mathcal{L}(f_{\bw}(g\cdot\bar{\bx}_i),Y)\}  \Big)\Big| \nonumber 
    \\  
    & \qquad + (L_{\mathcal{L}}  + C L_{p} L_{\max}) \varepsilon
    \label{eq:thm_eq_3a}
\end{align}
where $N_i = \# \{j : [\bar{\bx}_i]_G \text{ is closest to } [\bx_j]_G\}$, and the last step comes from the coreset coverage condition.

From the representative coreset definition in Subsection~\ref{subsec:slim-learning}, the coreset weights can be expressed in terms of $N_i$ as $w_i = N_i/N$. This allows us to bound the first term in \eqref{eq:thm_eq_3a} as
\begin{align}
    & \bbE_{\bbQ(g)}\{\mathcal{L}(f_{\bw}(g\cdot\bar{\bx}_i),y_i)\} - \bbE_{p(Y|g\cdot \bar{\bx}_i)\bbQ(g)}\{\mathcal{L}(f_{\bw}(g\cdot\bar{\bx}_i),Y)\} \nonumber\\
    & = \bbE_{\bbQ(g)}\Big\{\sum_{c\in\mathcal{Y}} \big(\delta_{c-y_i} - p(Y=c|g\cdot \bar{\bx}_i)\big) \mathcal{L}(f_{\bw}(g\cdot\bar{\bx}_i),c) \Big\} \nonumber \\
    & = \bbE_{\bbQ(g)}\Big\{\sum_{c\in\mathcal{Y}} p(Y=c|g\cdot \bar{\bx}_i) \big(\mathcal{L}(f_{\bw}(g\cdot\bar{\bx}_i),y_i) - \mathcal{L}(f_{\bw}(g\cdot\bar{\bx}_i),c) \big) \Big\} \nonumber \\
    & \le
        L_{\max} \bbE_{\bbQ(g)}\{ \sum_{c\in\mathcal{Y}} |\delta_{c-y_i} - p(Y=c|g\cdot \bar{\bx}_i)|\} 
    \\
       & = 2(1-p(Y=y_i| \bar{\bx}_i)) L_{\max}
    \nonumber \\
    & = 2 L_{\max} \eta_i
    \label{eq:thm_eq_4a}
\end{align}
where $\eta_i:= 1-p(Y=y_i| \bar{\bx}_i)$ is the label uncertainty in the $i$-th coreset sample.
Thus, we can bound \eqref{eq:thm_eq_3a} by
\begin{align}
    \leq 2 L_{\max} \sum_{i=1}^{K} w_i \eta_i + (L_{\mathcal{L}}  + C L_{p} L_{\max}) \varepsilon \,.
    \label{eq:thm_eq_5a}
\end{align}
\textbf{Concluding:} By combining the bounds in \eqref{eq:thm_bernstein_2} and \eqref{eq:thm_eq_5a} we obtain the desired result.
\end{proof}

\end{document}